\documentclass{article}

    \RequirePackage{fullpage}
    \usepackage[utf8]{inputenc} 
    \usepackage[T1]{fontenc}    
    \usepackage{lmodern} 

\usepackage[utf8]{inputenc} 
\usepackage[T1]{fontenc}    
\usepackage{url}            
\usepackage{booktabs}       
\usepackage{amsfonts}       
\usepackage{nicefrac}       
\usepackage{microtype}      
\usepackage{xcolor}
\usepackage{mathtools}
\usepackage{etoolbox}
\usepackage{comment}
\usepackage{amsmath}
\usepackage{hyperref}       
\hypersetup{colorlinks,linkcolor=blue,citecolor=blue,urlcolor=blue}
\usepackage{bbm}
\usepackage{tikz}
\usepackage{enumitem}
\setlist{leftmargin=5.5mm}
\usepackage{tkz-euclide}
\colorlet{light blue}{blue!20}
\usepackage{tikz-qtree}
\usepackage{istgame}
\usepackage{subcaption}
\usepackage{multirow}
\usepackage{algorithm}
\usepackage[noend]{algpseudocode}
\usepackage{amsthm}
\usepackage{thm-restate}
\usepackage{authblk}
\newtheorem{proposition}{Proposition}[section]

\declaretheorem[name=Lemma,numberwithin=section]{lemma}

\declaretheorem[name=Definition,numberwithin=section]{definition}
\declaretheorem[name=Remark,numberwithin=section]{remark}
\declaretheorem[name=Example,numberwithin=section]{example}
\newcommand{\name}[1]{\textsc{HypHC}}
\newtoggle{comment}
\toggletrue{comment}

\newcommand{\CITE}[1]{\textcolor{red}{[CITE]}}
\newcommand{\REF}[1]{\textcolor{red}{[REF]}}
\newcommand{\dec}{\mathsf{dec}}

\usepackage[capitalise]{cleveref}       

\begin{document}
\title{From Trees to Continuous Embeddings and Back:\\ Hyperbolic Hierarchical Clustering}
 \author[$\ddagger$]{Ines~Chami}
 \author[$\dagger$]{Albert Gu}
  \author[$\dagger\dagger$]{
Vaggos Chatziafratis}
 \author[$\dagger$]{Christopher~R{\'e}}
 \affil[$\dagger$]{Department of Computer Science, Stanford University}
 \affil[$\ddagger$]{Institute for Computational and Mathematical Engineering, Stanford University}
  \affil[$\dagger\dagger$]{Google Research, NY}
 \affil[ ]{\footnotesize{\texttt{\{chami, albertgu, vaggos, chrismre\}@cs.stanford.edu}}}

\maketitle
\makeatletter
\renewcommand\paragraph{\@startsection{paragraph}{4}{\z@}{0.15ex \@plus1ex \@minus.2ex}{-1em}{\normalfont\normalsize\bfseries}}
\makeatother

\begin{abstract}
Similarity-based Hierarchical Clustering (HC) is a classical unsupervised machine learning algorithm that has traditionally been solved with heuristic algorithms like Average-Linkage. 
Recently, Dasgupta~\cite{dasgupta2016cost} reframed HC as a discrete optimization problem by introducing a global cost function measuring the quality of a given tree. 
In this work, we provide the first continuous relaxation of Dasgupta's discrete optimization problem with provable quality guarantees. 
The key idea of our method, \name{}, is showing a direct correspondence from discrete trees to continuous representations (via the hyperbolic embeddings of their leaf nodes) and back (via a decoding algorithm that maps leaf embeddings to a dendrogram), allowing us to search the space of discrete binary trees with continuous optimization.
Building on analogies between trees and hyperbolic space, we derive a continuous analogue for the notion of lowest common ancestor, which leads to a continuous relaxation of Dasgupta's discrete objective.
We can show that after decoding, the global minimizer of our continuous relaxation yields a discrete tree with a $(1+\varepsilon)$-factor approximation for Dasgupta's optimal tree, where $\varepsilon$ can be made arbitrarily small and controls optimization challenges. 
We experimentally evaluate \name{} on a variety of HC benchmarks and find that even approximate solutions found with gradient descent have superior clustering quality than agglomerative heuristics or other gradient based algorithms.
Finally, we highlight the flexibility
of \name{} using end-to-end training in a downstream classification task.

\end{abstract}

\section{Introduction}
Hierarchical Clustering (HC) is a fundamental problem in data analysis, where given datapoints and their pairwise similarities, the goal is to construct a hierarchy over clusters, in the form of a tree whose leaves correspond to datapoints and internal nodes correspond to clusters. HC naturally arises in standard applications where data exhibits hierarchical structure, ranging from phylogenetics~\cite{felsenstein2004inferring} and cancer gene sequencing~\cite{sorlie2001gene, sotiriou2003breast} to text/image analysis~\cite{steinbach2000}, community detection~\cite{leskovec2019mining} and everything in between. A family of easy to implement, yet slow, algorithms includes agglomerative methods (e.g., Average-Linkage) that build the tree in a bottom-up manner by iteratively merging pairs of similar datapoints or clusters together. In contrast to ``flat'' clustering techniques like $k$-means, HC provides fine-grained interpretability and rich cluster information at all levels of granularity and alleviates the requirement of specifying a fixed number of clusters a priori.

Despite the abundance of HC algorithms, the HC theory was underdeveloped, since no “global” objective function was associated with the final tree output.
A well-formulated objective allows us to compare different algorithms, measure their quality, and explain their success or failure.\footnote{Contrast this lack of HC objectives with flat clustering, where $k$-means objectives and others have been studied intensively from as early as 1960s (e.g.,~\cite{hartigan1975clustering}), leading today to a comprehensive theory on clustering.} 
To address this issue, Dasgupta~\cite{dasgupta2016cost} recently introduced a discrete cost function over the space of binary trees with $n$ leaves. 
A key property of his cost function is that low-cost trees correspond to meaningful hierarchical partitions in the data.
He initially proved this for symmetric stochastic block models, and later works provided experimental evidence~\cite{roy2017hierarchical}, or showed it for hierarchical stochastic block models, suitable for inputs that contain a ground-truth hierarchical structure~\cite{cohen2017hierarchical,cohen2019hierarchical}. 
These works led to important steps towards understanding old and building new HC algorithms~\cite{charikar2017approximate,charikar2019hierarchical,moseley2017approximation,alon2020hierarchical}.

The goal of this paper is to improve the performance of HC algorithms using a differentiable relaxation of Dasgupta's discrete optimization problem. 
There have been recent attempts at gradient-based HC via embedding methods, which do not directly relax Dasgupta's optimization problem.
UFit~\cite{chierchia2019ultrametric} addresses a different ``ultrametric fitting'' problem using Euclidean embeddings, while gHHC~\cite{monath2019gradient} assumes that partial information about the optimal clustering---more specifically leaves' hyperbolic embeddings---is known.
Further, these two approaches lack theoretical guarantees in terms of clustering quality and are outperformed by discrete agglomerative algorithms.

A relaxation of Dasgupta's discrete optimization combined with the powerful toolbox of gradient-based optimization has the potential to yield improvements in terms of (a) \textbf{clustering quality} (both theoretically and empirically), (b) \textbf{scalability}, and (c) \textbf{flexibility}, since a gradient-based approach can be integrated into machine learning (ML) pipelines with end-to-end training. 
However, due to the inherent discreteness of the HC optimization problem, several challenges arise:
\begin{enumerate}[label=(\arabic*),topsep=0pt]
    \itemsep -2pt
    \item How can we continuously \textbf{parameterize the search space} of discrete binary trees?
    A promising direction is leveraging \emph{hyperbolic embeddings} which are more aligned with the geometry of trees than standard Euclidean embeddings~\cite{sarkar2011low}. 
    However, hyperbolic embedding methods typically assume a \emph{fully known}~\cite{nickel2017poincare,sala2018representation} or partially known graph~\cite{monath2019gradient} that will be embedded, whereas the challenge here is searching over an exponentially large space of trees with \emph{unknown} structure.
    \item How can we derive a \textbf{differentiable relaxation} of the HC cost? One of the key challenges is that this cost relies on discrete properties of trees such as the {lowest common ancestor} (LCA). 
    \item How can we \textbf{decode} a discrete binary tree from continuous representations, while ensuring that the ultimate discrete cost is close to the continuous relaxation?
\end{enumerate}
Here, we introduce \name{}, an end-to-end differentiable model for HC with provable guarantees in terms of clustering quality, which can be easily incorporated into ML pipelines. 
\begin{enumerate}[label=(\arabic*),topsep=0pt]
    \itemsep -2pt
    \item Rather than minimizing the cost function by optimizing over discrete trees, we parameterize trees using leaves' hyperbolic embeddings. 
    In contrast with Euclidean embeddings, hyperbolic embeddings can represent trees with arbitrarily low distortion in just two dimensions~\cite{sarkar2011low}.
    We show that the leaves themselves provide enough information about the underlying tree, avoiding the need to explicitly represent the discrete structure of internal nodes.
    \item We derive a {continuous} LCA analogue,
    which leverages the analogies between shortest paths in trees and hyperbolic geodesics~(\cref{fig:lca}),
    and propose a differentiable variant of Dasgupta's cost.
    \item  We propose a decoding algorithm for the internal nodes which maps the learned leaf embeddings to a dendrogram (cluster tree) of low distortion.
\end{enumerate}
We show (a) that our approach produces good \textbf{clustering quality}, in terms of Dasgupta cost.
Theoretically, assuming perfect optimization, the optimal clustering found using \name{} yields a $(1+\varepsilon)$-approximation to the minimizer of the discrete cost, where $\varepsilon$ can be made arbitrarily small, and controls the tradeoffs between quality guarantees and optimization challenges.
{Notice that due to our perfect optimization assumption, this does not contradict previously known computational hardness results based on the Small Set Expansion for achieving constant approximations~\cite{charikar2017approximate}, but allows us to leverage the powerful toolbox of nonconvex optimization.}
Empirically, we find that even approximate \name{} solutions found with gradient descent outperform or match the performance of the best discrete and continuous methods on a variety of HC benchmarks. 
Additionally, (b) we propose two extensions of \name{} that enable us to \textbf{scale} to large inputs and we also study the tradeoffs between clustering quality and speed.
Finally, (c) we demonstrate the \textbf{flexibility} of \name{} by jointly optimizing embeddings for HC and a downstream classification task, improving accuracy by up to 3.8\% compared to the two-stage embed-then-classify approach.

\begin{figure}
     \centering
     \begin{subfigure}[b]{0.24\textwidth}
     \resizebox{\textwidth}{!}{\renewcommand{\arraystretch}{0.95}
     \input{figures/tikz_tree}}
\caption{}\label{fig:triplet_cond}
     \end{subfigure}
     \begin{subfigure}[b]{0.24\textwidth}
         \centering
         \includegraphics[width=0.9\textwidth]{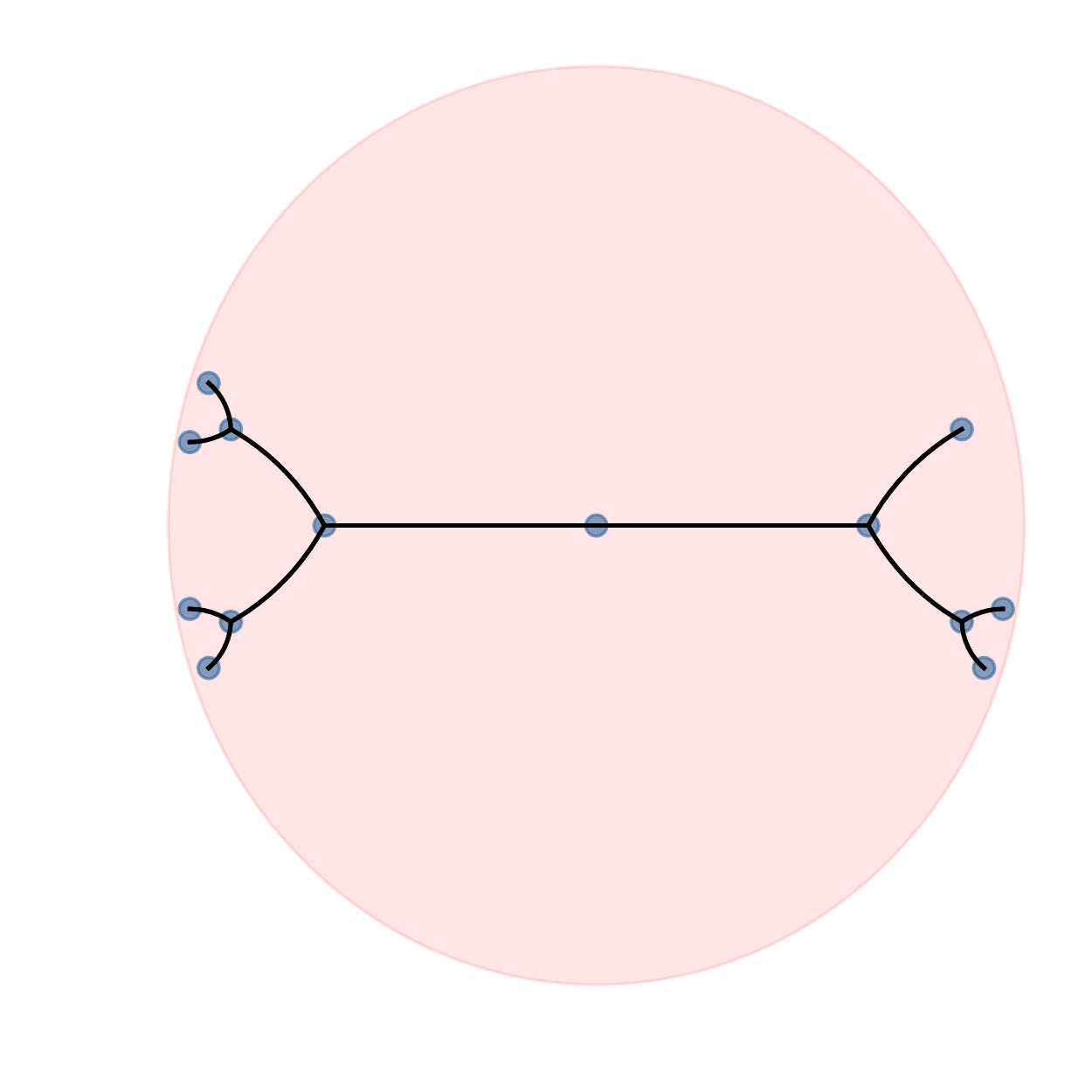}
         \caption{}\label{fig:sarkar}
     \end{subfigure}
     \begin{subfigure}[b]{0.24\textwidth}
         \centering
         \includegraphics[width=0.9\textwidth]{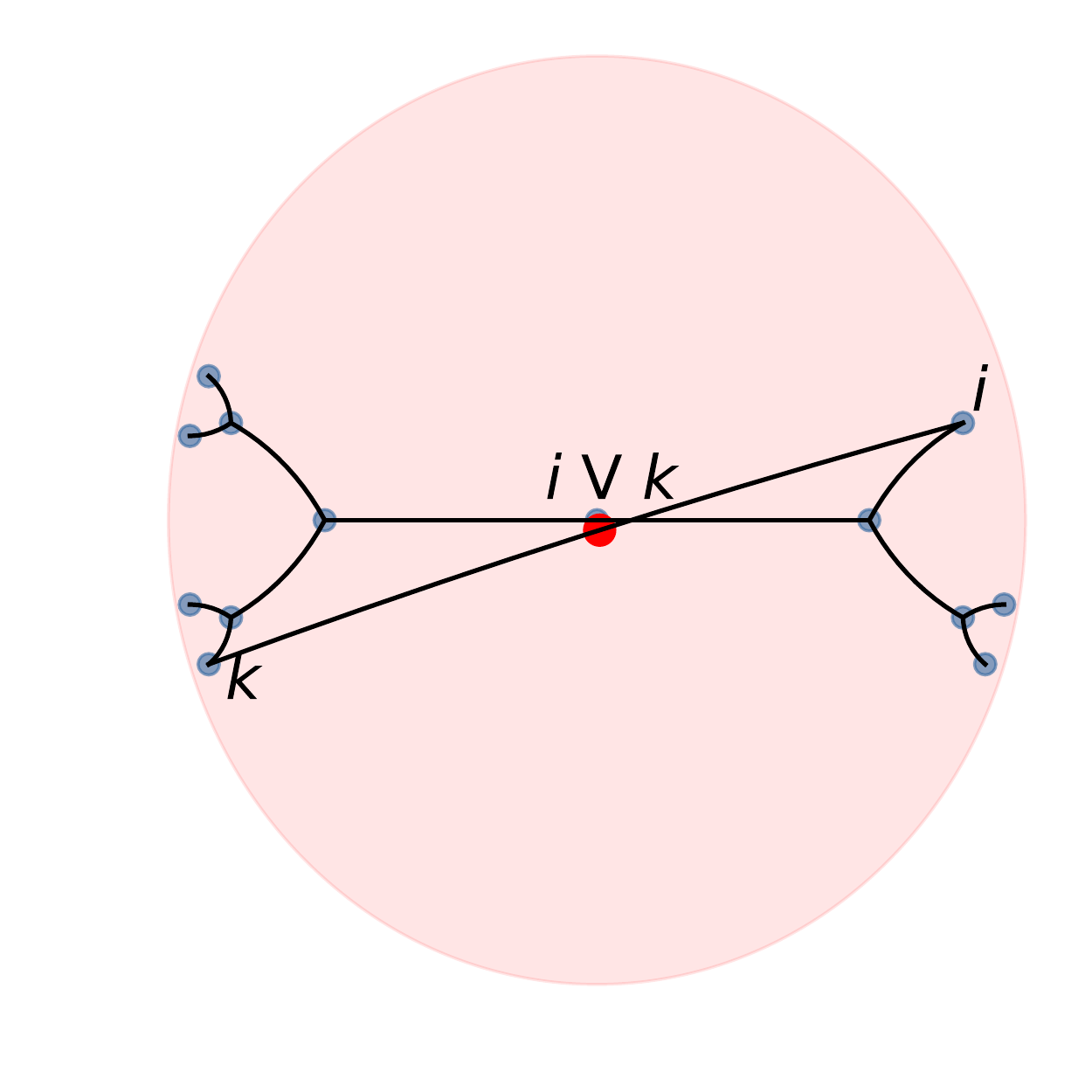}
         \caption{}
         \label{fig:hyp_lca_1}
     \end{subfigure}
     \begin{subfigure}[b]{0.24\textwidth}
         \centering
         \includegraphics[width=0.9\textwidth]{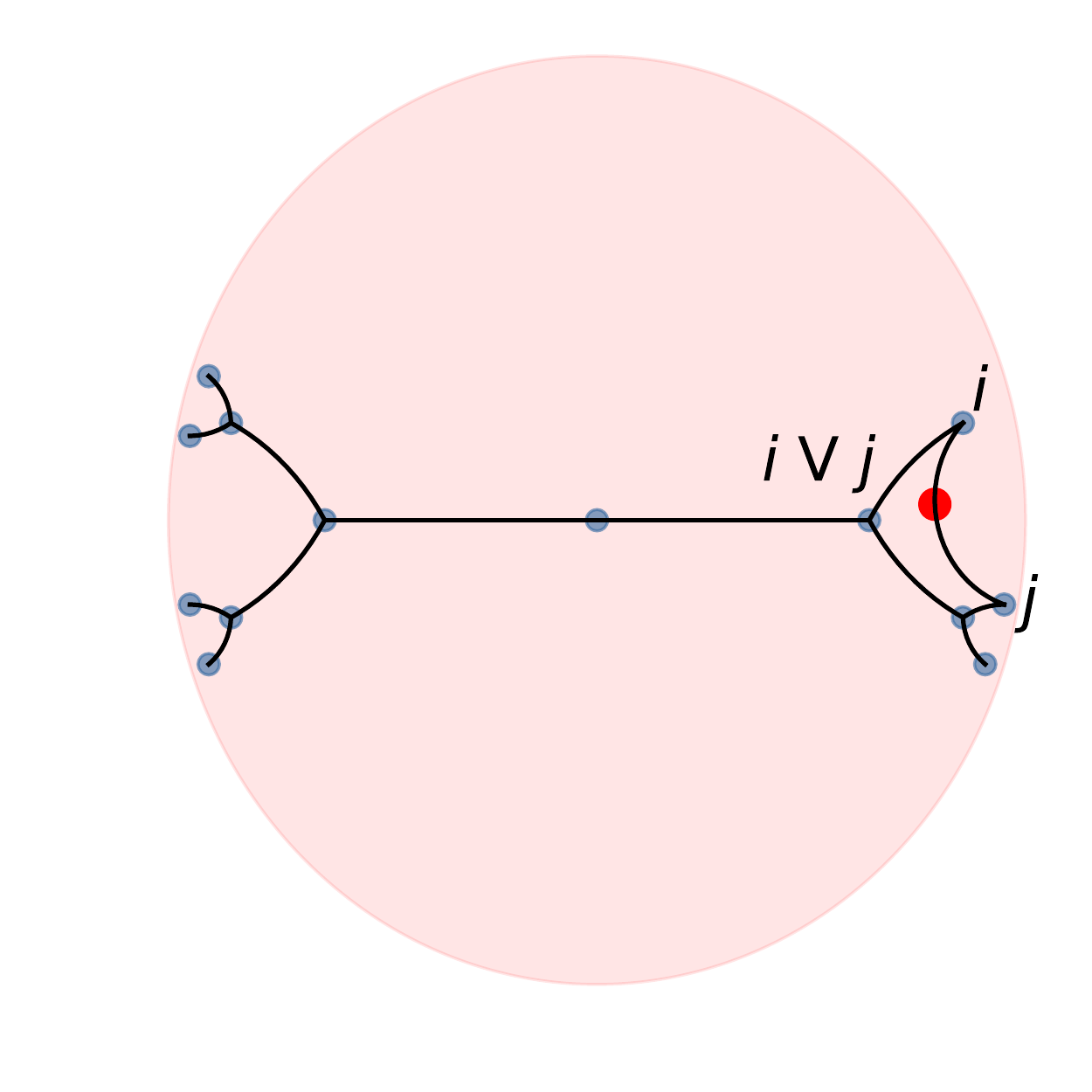}
         \caption{}
         \label{fig:hyp_lca_2}
     \end{subfigure}
     \caption{
     The tree shown in (a) is embedded into hyperbolic space ($\mathbb{B}_2$) in (b), (c), and (d). 
     In (c) and (d) we show the hyperbolic LCA (in red) and illustrate the relationship between the discrete LCA, which is central to Dasgupta's cost, and geodesics (shortest paths) in hyperbolic space. 
     }\label{fig:lca}
\end{figure}
\section{Related Work}
\paragraph{Gradient-based Hierarchical Clustering}
Chierchia et al.~\cite{chierchia2019ultrametric} introduce an ultra-metric fitting framework (UFit) to learn embeddings in ultra-metric spaces.
These are restrictive metric spaces (and a special case of tree metrics) where the triangle inequality is strengthened to $d(x, z)\le\mathrm{max}\{d(x, y), d(y, z)\}$.
UFit can be extended to HC using an agglomerative method to produce a dendrogram.
gHHC~\cite{monath2019gradient}, which inspired our work, is the first attempt at HC with hyperbolic embeddings.
This model assumes input leaves' hyperbolic embeddings are given (by fixing their distance to the origin) and optimizes embeddings of internal nodes.
Observing that the internal tree structure can be directly inferred from leaves' hyperbolic embeddings, we propose an approach that is a direct relaxation of Dasgupta's discrete optimization framework and does not assume any information about the optimal clustering. 
In contrast with previous gradient-based approaches, our approach has theoretical guarantees in terms of clustering quality and empirically outperforms agglomerative heuristics. 

\paragraph{Dasgupta's cost}
Having an objective function for HC is crucial not only for guiding the optimization, but also for having theoretically grounded solutions. Unfortunately, there has been a lack of global objectives measuring the quality of a HC, which is in stark contrast with the plethora of flat clustering objectives like $k$-means, $k$-center, correlation clustering (e.g.,~\cite{charikar2005clustering, bansal2004correlation}) and many more. This lack of optimization objectives for HC was initially addressed indirectly in~\cite{dasgupta2002performance} (by comparing a tree against solutions to $k$-clustering objectives for multiple $k$), and has again emerged recently by Dasgupta~\cite{dasgupta2016cost}, who introduced a cost to evaluate and compare the performance of HC algorithms. This discrete HC cost has the key property that good clustering trees should yield a low cost. The formulation of this objective favors binary trees and the optimum solution can always be assumed to be binary, as a tree with higher fan-out in any internal node can easily be modified into a binary tree with the same cost (or less).
In particular, one interesting aspect of this objective is that running Recursive Sparsest Cut would produce a HC with provable guarantees with respect to his cost function~\cite{charikar2017approximate}. 
Subsequent work was able to shed light to Average Linkage performance; specifically, \cite{moseley2017approximation} studied the complement to Dasgupta's cost function and showed that Average Linkage will find a solution within 33\% of the optimizer. 
Further techniques based on semidefinite programming relaxations led to improved approximations~\cite{charikar2019hierarchical,ahmadian2019bisect,alon2020hierarchical}, however with a significant overhead in runtime (see~\cite{chatziafratis2020hierarchical} for a survey).

\paragraph{Agglomerative Hierachical Clustering}
One of the first suite of algorithms developed to solve HC was bottom-up linkage methods. These are simple and easy to implement algorithms that recursively merge similar datapoints to form some small clusters and then gradually larger and larger clusters emerge. Well-known heuristics include Single Linkage, Complete Linkage and Average Linkage, that we briefly describe here. All three heuristics start with $n$ datapoints forming singleton clusters initially and perform exactly $n-1$ merges in total until they produce a binary tree corresponding to the HC output. At any given step, if $A$ and $B$ denote two already formed clusters, the criterion for which clusters to merge next is to \textit{minimize} the minimum, maximum and average distance between two clusters $A$ and $B$, for Single, Complete and Average Linkage respectively.  
If instead of pairwise distances the input was given as a similarity graph, analogous formulas can be used, where the criterion is to maximize the respective quantities. These algorithms can be made to run in time $O(n^2\log n)$. 
A bottleneck at the core of the computations for such algorithms is the nearest neighbor problem, for which known almost quadratic lower bounds apply~\cite{cs2018closest,alman2015probabilistic}. However, recent works~\cite{abboud2019subquadratic,charikar2019euclidean} try to address this in Euclidean spaces where we are provided features by using locality sensitive hashing techniques and approximate nearest neighbors. 
Instead, here, we learn these features using similarities only, and then use a fast decoding algorithm on the learned features to output the final tree. Finally, HC has been studied in the model of parallel computation, where a variation of Boruvka's minimum spanning tree algorithm was used~\cite{bateni2017affinity}.

\paragraph{Hyperbolic embeddings}
Hyperbolic geometry has been deeply studied in the network science community~\cite{krioukov2010hyperbolic,papadopoulos2012popularity}, with applications to network routing for instance~\cite{cvetkovski2009hyperbolic,kleinberg2007geographic}.
Recently, there has been interest in using hyperbolic space to embed data that exhibits hierarchical structures. 
In particular, Sarkar~\cite{sarkar2011low} introduced a combinatorial construction that can embed trees with high fidelity, in just two dimensions~(\cref{fig:sarkar}). 
Sala et al.~\cite{sala2018representation} extend this construction to higher dimensions and study the precision-quality tradeoffs that arise with hyperbolic embeddings. 
Nickel et al.~\cite{nickel2017poincare} propose a gradient-based method to embed taxonomies in the Poincar\'e model of hyperbolic space~\cite{nickel2017poincare}, which was further extended to the hyperboloid model~\cite{nickel2018learning}.
More recently, hyperbolic embeddings have been applied to neural networks~\cite{ganea2018hyperbolic} and later extended to graph neural networks~\cite{chami2019hyperbolic,liu2019hyperbolic} and knowledge graph embeddings~\cite{chami-etal-2020-low} (see~\cite{chami2020machine} for a survey).
In contrast with all the above methods that assume a known input graph structure, \name{} discovers and decodes a tree structure using pairwise similarities.

\section{Background}
We first introduce our notations and the problem setup.
We then briefly review basic notions from hyperbolic geometry and refer to standard texts for a more in-depth treatment~\cite{carmo1992riemannian}.
\cref{appendix:preliminaries} includes additional background about hyperbolic geometry for our technical proofs.
\subsection{Similarity-Based Hierarchical clustering}
\paragraph{Notations} 
We consider a dataset $\mathcal{D}$ with $n$ datapoints and pairwise similarities $(w_{ij})_{i, j \in [n]}$.
A HC of $\mathcal{D}$ is a rooted binary tree $T$ with exactly $n$ leaves, such that each leaf corresponds to a datapoint in $\mathcal{D}$, and intermediate nodes represent clusters. 
For two leaves $(i, j)$ in $T$, we denote $i\vee j$ their LCA, and $T[i\vee j]$ the subtree rooted at $i\vee j$, which represents the smallest cluster containing both $i$ and $j$.
In particular, we denote $\mathrm{leaves}(T[i\vee j])$ the leaves of $T[i\vee j]$, which correspond to datapoints in the cluster $i\vee j$.
Finally, given three leaf nodes $(i, j, k)$ in $T$, we say that the relation $\{i, j|k\}$ holds if $i\vee j$ is a descendant of $i\vee j\vee k$~(\cref{fig:triplet_cond}).
In this scenario, we have $i\vee k=j\vee k$.

\paragraph{HC discrete optimization framework}
The goal of HC is to find a binary tree $T$ that is mindful of pairwise similarities in the data. 
Dasgupta~\cite{dasgupta2016cost} introduced a cost function over possible binary trees with the crucial property that good trees should have a low cost: $C_{\mathrm{Dasgupta}}(T; w) = \sum_{ij}w_{ij}|\mathrm{leaves}(T[i\vee j])|$.
Intuitively, a good tree for Dasgupta's cost merges similar nodes (high $w_{ij}$) first in the hierarchy (small subtree).
In other words, a good tree should cluster the data such that similar datapoints have LCAs further from the root than dissimilar datapoints. As the quantity $|\mathrm{leaves}(T[i\vee j])|$ may be hard to parse, we note here a simpler way to think of this term based on triplets of datapoints $i,j,k$: whenever a tree splits $i,j,k$ for the first time, then $k\notin \mathrm{leaves}(T[i\vee j])$ if and only if $k$ was separated from $i,j$. This triplet perspective has been fruitful in approximation algorithms for constrained HC~\cite{chatziafratis2018hierarchical}, in the analysis of semidefinite programs~\cite{charikar2019hierarchical} and for formulating normalized variants of Dasgupta's objective~\cite{wang2018improved}. 

More formally, Wang and Wang~\cite{wang2018improved} observe that:
\begin{equation}\label{eq:dasgupta_wang}
\begin{split}
    C_{\mathrm{Dasgupta}}(T; w) 
    &=\sum_{ijk}[w_{ij} + w_{ik} + w_{jk}-w_{ijk}(T; w)] + 2 \sum_{ij}w_{ij}\\
    \text{where}\ \ w_{ijk}(T; w) & = w_{ij}\mathbbm{1}[\{i, j|k\}] + w_{ik}\mathbbm{1}[\{i, k|j\}] + w_{jk}\mathbbm{1}[\{j, k|i\}],
\end{split}
\end{equation}
which reduces difficult-to-compute quantities about subtree sizes to simpler statements about LCAs.
For binary trees, exactly one of $\mathbbm{1}[\{i, j|k\}], \mathbbm{1}[\{i, k|j\}], \mathbbm{1}[\{j, k|i\}]$ holds, and these are defined through the notion of LCA. 
As we shall see next, we can relax this objective using a continuous notion of LCA in hyperbolic space. 
Dasgupta's cost-based perspective leads to a natural optimization framework for HC, where the goal is to find $T^*$ such that:
\begin{equation}\label{eq:opt}
    T^* = \underset{\mathrm{all}\ \mathrm{binary}\ \mathrm{trees}\ T}{\mathrm{argmin}}\  C_{\mathrm{Dasgupta}}(T; w).
\end{equation}

\subsection{Hyperbolic geometry}
\paragraph{Poincar\'e model of hyperbolic space}
Hyperbolic geometry is a non-Euclidean geometry with a constant negative curvature. 
We work with the Poincar\'e model with negative curvature $-1$: $\mathbb{B}_d=\{x\in\mathbb{R}^d:||x||_2\le 1\}$.
Curvature measures how an object deviates from a flat surface;  small absolute curvature values recover Euclidean geometry, while negative curvatures become ``tree-like''.

\paragraph{Distance function and geodesics}
The hyperbolic distance between two points $(x, y)\in(\mathbb{B}_d)^2$ is:
\begin{equation}\label{eq:hyp_distance}
    d(x, y)=\mathrm{cosh}^{-1}\bigg(1+2\frac{||x-y||_2^2}{(1-||x||_2^2)(1-||y||_2^2)}\bigg).
\end{equation}
If $y=o$, the origin of the hyperbolic space, the distance function has the simple expression: $d_o(x)\coloneqq d(o, x)=2\ \mathrm{tanh}^{-1}(||x||_2)$.
This distance function induces two types of geodesics: straight lines that go through the origin, and segments of circles perpendicular to the boundary of the ball.
These geodesics resemble shortest paths in trees: the hyperbolic distance between two points approaches the sum of distances between the points and the origin, similar to trees, where the shortest path between two nodes goes through their LCA~(\cref{fig:lca}).

\section{Hyperbolic Hierarchical Clustering}
\begin{figure}[t]
    \centering
    \begin{subfigure}[b]{0.194\textwidth}
    \includegraphics[width=\textwidth]{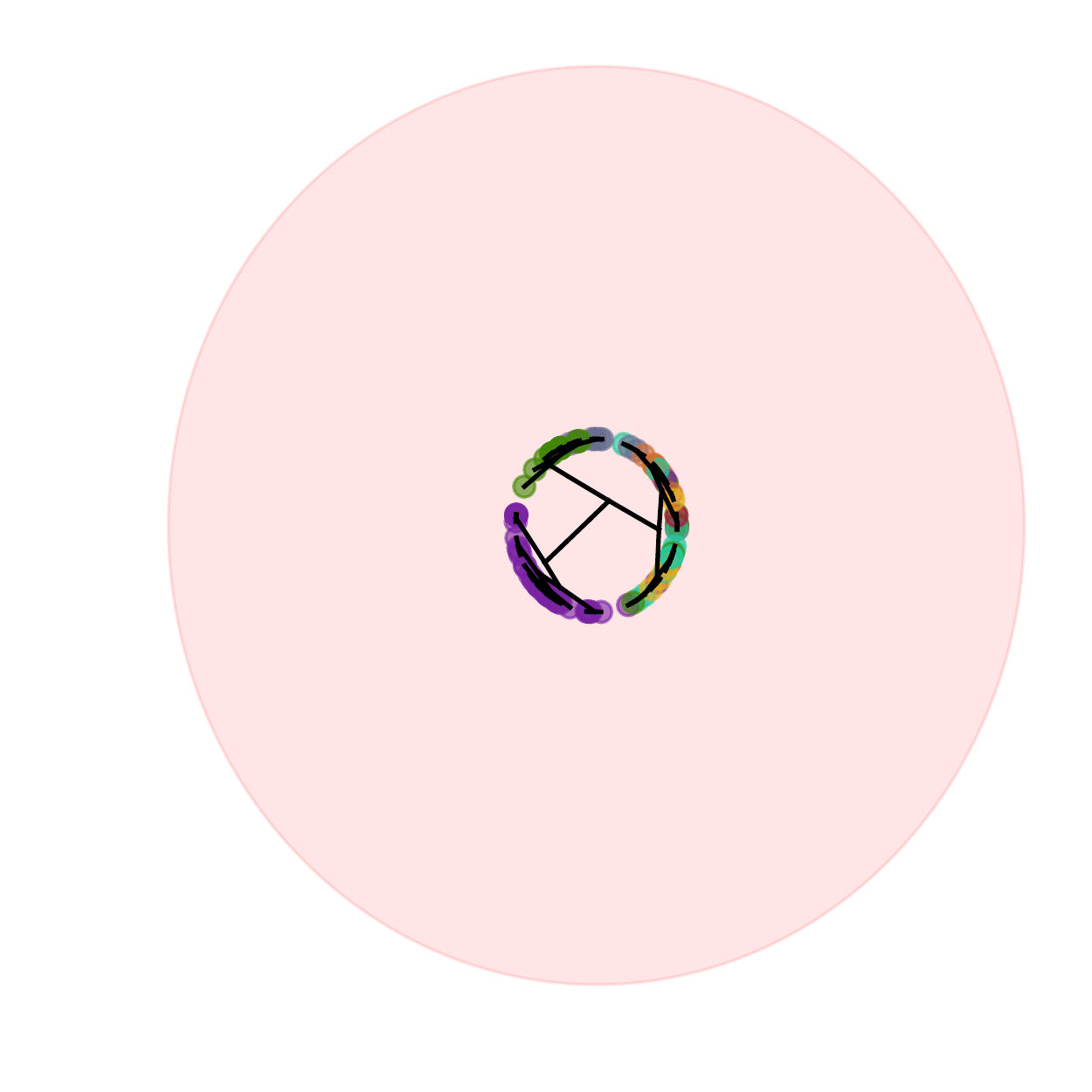}
    \end{subfigure}
    \begin{subfigure}[b]{0.194\textwidth}
    \includegraphics[width=\textwidth]{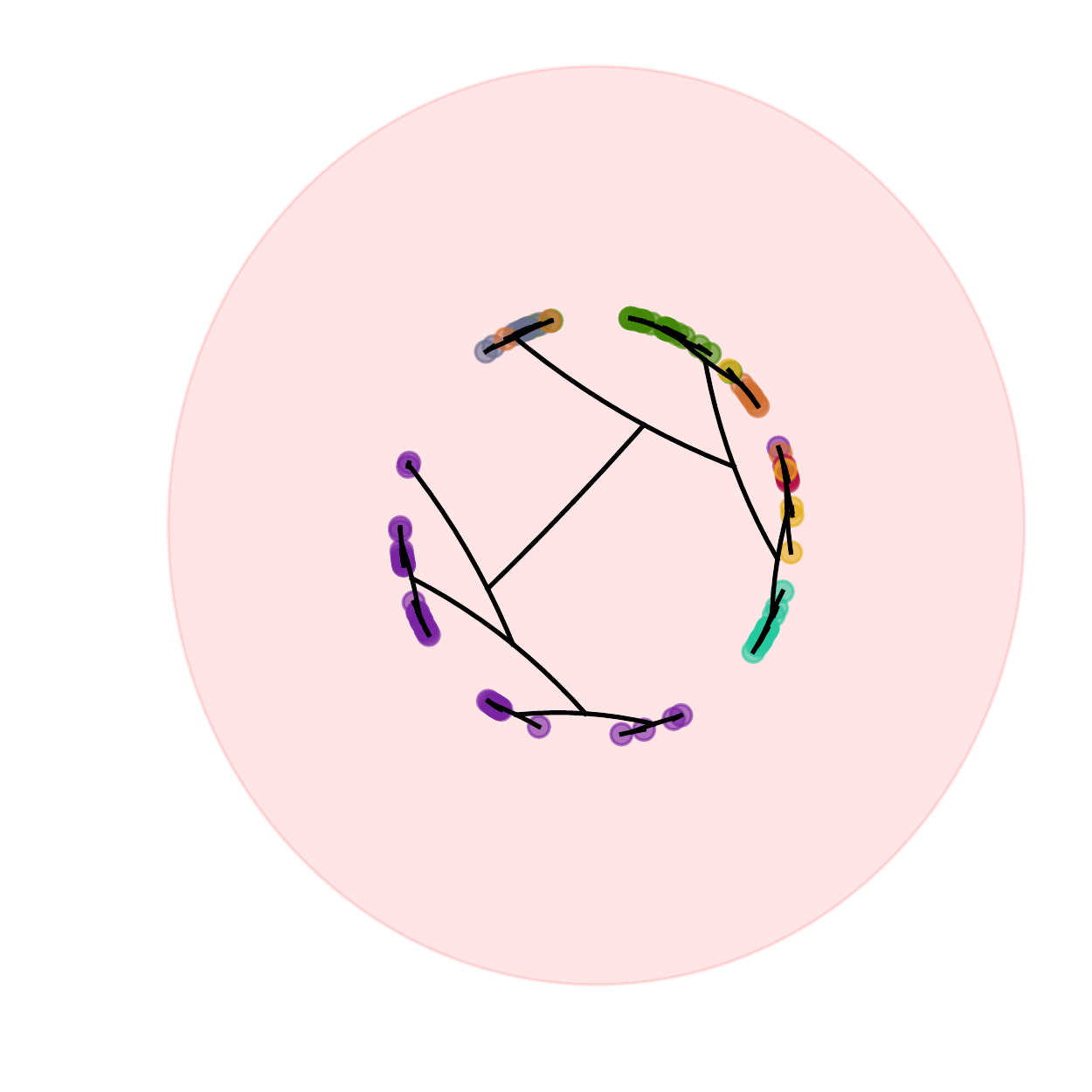}
    \end{subfigure}
    \begin{subfigure}[b]{0.194\textwidth}
    \includegraphics[width=\textwidth]{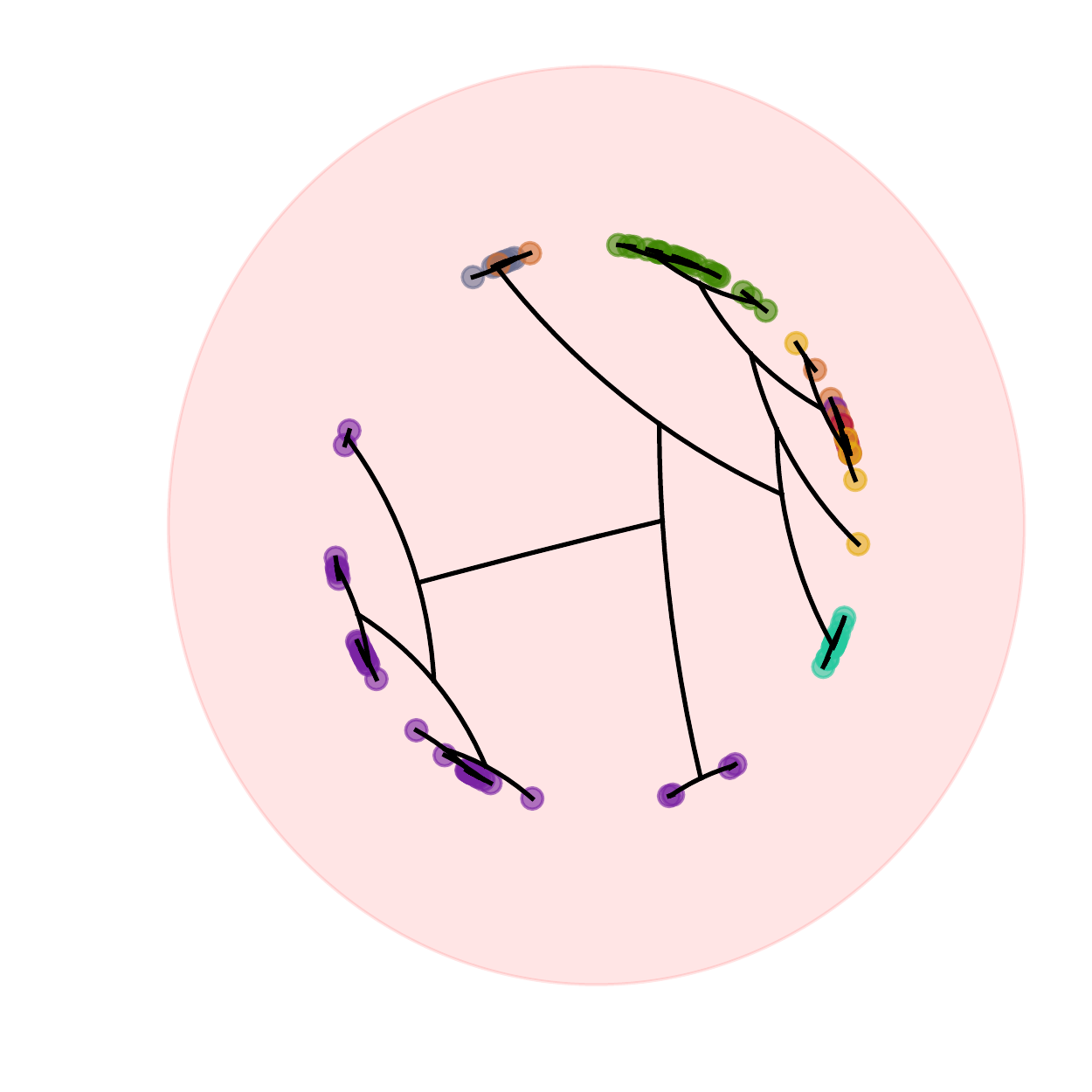}
    \end{subfigure}
    \begin{subfigure}[b]{0.194\textwidth}
    \includegraphics[width=\textwidth]{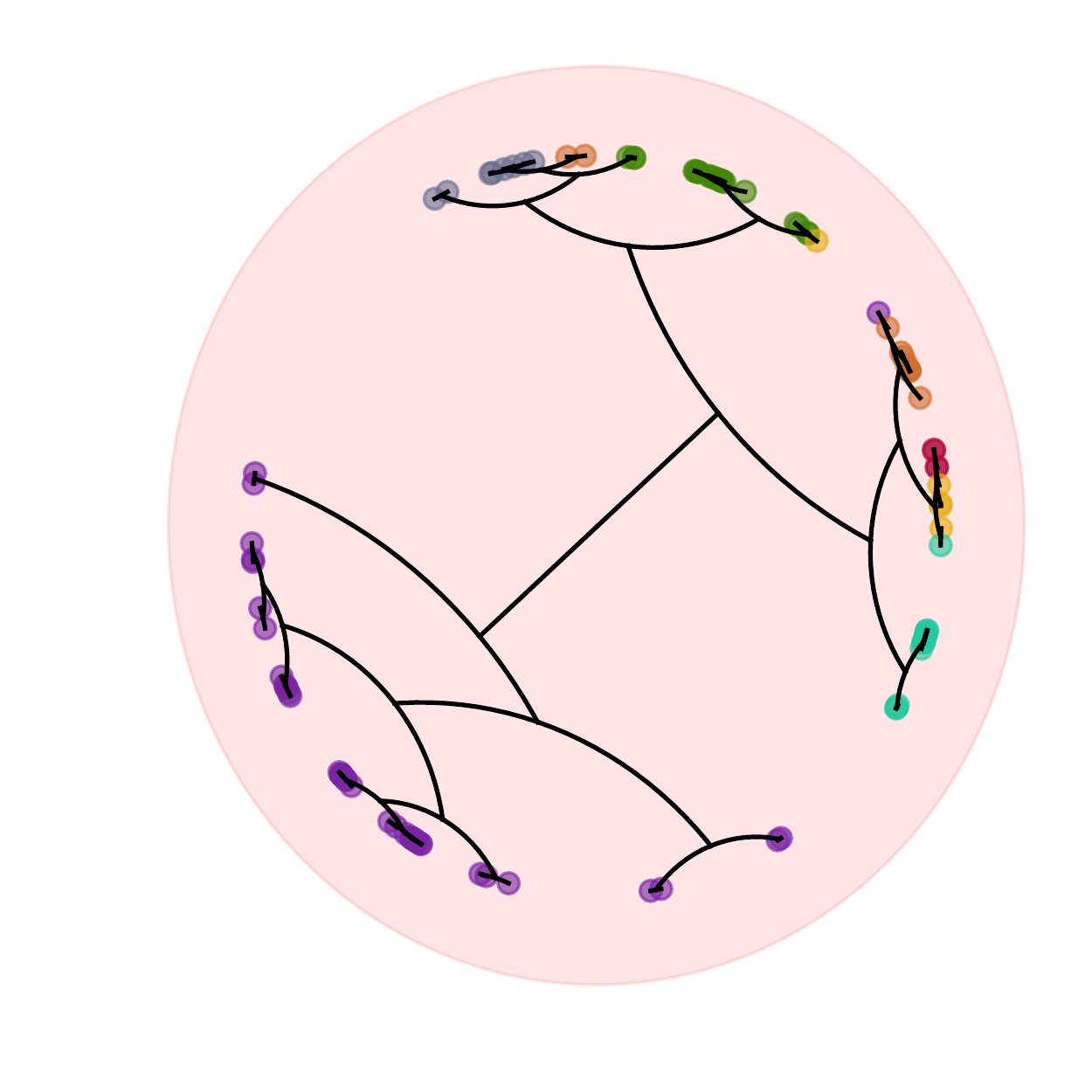}
    \end{subfigure}
    \begin{subfigure}[b]{0.194\textwidth}
    \includegraphics[width=\textwidth]{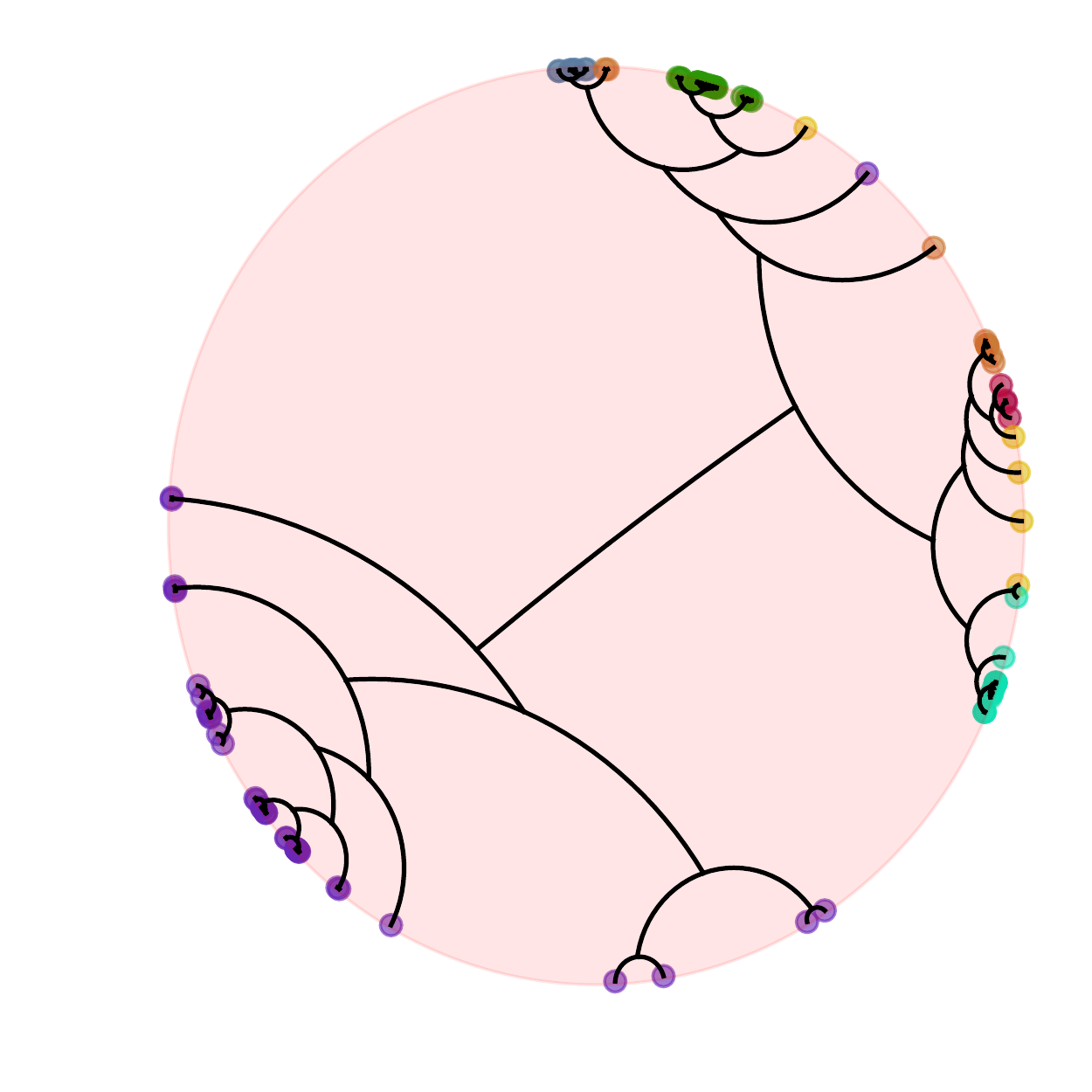}
    \end{subfigure}
    \caption{Visualization of \name{} embeddings and decoded trees during optimization.}\label{fig:scale}
\end{figure}
Our goal is to relax the discrete optimization problem in~\cref{eq:opt}. 
To do so, we represent trees using the hyperbolic embeddings of their leaves in the Poincar\'e disk (\cref{subsec:diff}).\footnote{We developed our theory in two dimensions. As shown in Sala et al.~\cite{sala2018representation}, this can present optimization difficulties due to precision requirements.
However, using Prop 3.1 of~\cite{sala2018representation}, we can increase the dimension to reduce this effect.}
We introduce a differentiable HC cost, using a continuous LCA analogue (\cref{subsec:cont}), and then propose a decoding algorithm, $\dec(\cdot)$, which maps embeddings to discrete binary trees~(\cref{subsec:dec}).
Our main result is that, when embeddings are optimized over a special set of \textit{spread} embeddings (\cref{def:spread}), solving our continuous relaxation yields a $(1+\varepsilon)$-approximation to the minimizer of Dasgupta's cost (\cref{subsec:theory}), where $\varepsilon$ can be made arbitrarily small.
Next, we present our continuous optimization framework and then detail the different components of \name{}.
\paragraph{\name{} continuous optimization framework}
Formally, if $\mathcal{Z}\subset\mathbb{B}_2^n$ denotes an arbitrary constrained set of embeddings, we propose to optimize the following continuous constrained optimization problem as a proxy for the discrete problem in~\cref{eq:opt}:
\begin{equation}\label{eq:diff}
\begin{split}
    Z^*={\mathrm{argmin}}_{Z\in\mathcal{Z}}\  C_{\mathrm{\name{}}}(Z; w,\tau)\ 
    \mathrm{and}\ T=\dec(Z^*)
\end{split}
\end{equation}
where $C_\mathrm{\name{}}(\cdot;w,\tau)$ is our differentiable cost and $\tau$ is a temperature used to relax the $\max$ function.

\subsection{Continuous tree representation via hyperbolic embeddings}\label{subsec:diff}
To perform gradient-based HC, one needs a continuous parameterization of possible binary trees on a fixed number of leaves.
We parameterize binary trees using the (learned) hyperbolic embedding of their leaves. 
Our insight is that because hyperbolic space induces a close correspondence between embeddings and tree metrics~\cite{sarkar2011low}, the leaf embeddings alone provide enough information to recover the full tree.
In fact, as hyperbolic points are pushed towards the boundary, the embedding becomes more tree-like (\cref{fig:scale})~\cite{sala2018representation}.
Thus, rather than optimize over all possible tree structures, which could lead to combinatorial explosion, we embed only the leaf nodes in the Poincar\'e disk using an embedding look-up:
\begin{equation}
    i\rightarrow z_i\in\mathbb{B}_2.
\end{equation}
These leaves' embeddings are optimized with gradient-descent using a continuous relaxation of Dasgupta's cost~(\cref{subsec:cont}), and are then decoded into a discrete binary tree to produce a HC on the leaves~(\cref{subsec:dec}). 

\subsection{Differentiable objective function}\label{subsec:cont}
Having a continuous representation of trees is not sufficient for gradient-based HC, since Dasgupta's cost requires computing the discrete LCA. 
We leverage the similarities between geodesics in hyperbolic space and shortest paths in trees, to derive a continuous analogue of the discrete LCA.
We then use this hyperbolic LCA to introduce a differentiable version of Dasgupta's cost, $C_\name{}(\mathrm{\cdot};w,\tau)$.

\paragraph{Hyperbolic lowest common ancestor}
Given two leaf nodes $i$ and $j$, their LCA $i\vee j$ in $T$ is the node on the shortest path connecting $i$ and $j$ (denoted $i\leadsto j$) that is closest to the root node $r$: 
\begin{equation}
    i\vee j=\mathrm{argmin}_{k\in i\leadsto j}d_T(r, k).
\end{equation} 
Analogously, we define the hyperbolic LCA between two points $(z_i, z_j)$ in hyperbolic space as the point on their geodesic (shortest path, denoted $z_i\leadsto z_j$) that is closest to the origin (the root):
\begin{equation}\label{eq:hyp_lca}
    z_i \vee z_j\coloneqq\mathrm{argmin}_{z\in z_i\leadsto z_j}d(o, z).
\end{equation}
Note that $z_i \vee z_j$ is also the orthogonal projection of the origin onto the geodesic~(\cref{fig:lca}) and its hyperbolic distance to the origin can be computed exactly.
\begin{restatable}[]{lemma}{lcanorm}\label{lemma:lca_norm}
Let $(x, y)\in(\mathbb{B}_2)^2$ and $x\vee y$ denote the point on the geodesic connecting $x$ and $y$ that minimizes the distance to the origin $o$.
Let $\theta$ be the angle between $(x, y)$ and $\alpha$ be the angle between $(x, x\vee y)$. 
We have:
\begin{equation}\label{eq:lca_norm}
\begin{split}
    \alpha&=\mathrm{tan}^{-1}\bigg(\frac{1}{\mathrm{sin}(\theta)}\bigg(\frac{||x||_2(||y||_2^2+1)}{||y||_2(||x||_2^2+1)}-\mathrm{cos}(\theta)\bigg)\bigg),\\
    \mathrm{and}\ \ d_o(x\vee y)&=2\ \mathrm{tanh}^{-1}(\sqrt{R^2+1}-R),\\
    \mathrm{where}\ \ R&=\sqrt{\bigg(\frac{||x||_2^2+1}{2||x||_2\mathrm{cos}(\alpha)}\bigg)^2-1}.
\end{split}
\end{equation}
\end{restatable}
We provide a high-level overview of the proof and refer to~\cref{appendix:lca_construction} for a detailed derivation.
Geodesics in the Poincar\'e disk are straight lines that go through the origin or segments of circles that are orthogonal to the boundary of the disk. 
In particular, given two points, one can easily compute the coordinates of the circle that goes through the points and that is orthogonal to the boundary of the disk using circle inversions~\cite{brannan2011geometry}. 
One can then use this circle to recover the hyperbolic LCA, which is simply the point on the diameter that is closest to the origin (i.e. smallest Euclidean norm). 

\paragraph{\name{}'s objective function} 
The non-differentiable term in~\cref{eq:dasgupta_wang} is $w_{ijk}(T;w)$, which is the similarity of the pair that has the deepest LCA in $T$, i.e. the LCA that is farthest---in tree distance---from the root.
We use this qualitative interpretation to derive a continuous version of Dasgupta's cost.
Consider an embedding of $T$ with $n$ leaves, $Z=\{z_1, \ldots, z_n\}$. 
The notion of deepest LCA can be extended to continuous embeddings by looking at the continuous LCA that is farthest---in hyperbolic distance---from the origin. 
Our differentiable HC objective is then: 
\begin{equation}\label{eq:dasgupta_cont}
    \begin{split}
        C_{\mathrm{\name{}}}(Z; w,\tau) =\sum\nolimits_{ijk}&(w_{ij}+w_{ik}+w_{jk}-w_{\name{},{ijk}}(Z; w,\tau))+ 2 \sum\nolimits_{ij}w_{ij}\\
    \mathrm{where}\ w_{\name{},{ijk}}(Z; w,\tau)&=(w_{ij}, w_{ik}, w_{jk})\ \cdot\ \sigma_\tau(d_o( z_i \vee z_j), d_o( z_i \vee z_k), d_o(z_j \vee z_k))^\top,
    \end{split}
\end{equation}
and $\sigma_\tau(\cdot)$ is the scaled softmax function: $\sigma_\tau(\alpha)_i=e^{\alpha_i/\tau}/{\sum_je^{\alpha_j/\tau}}$.

\subsection{Hyperbolic decoding}\label{subsec:dec}
The output of HC needs to be a discrete binary tree, and optimizing the \name{} loss in~\cref{eq:diff} only produces leaves' embeddings.
We propose a way to decode a binary tree structure from embeddings by iteratively merging the most similar pairs based on their hyperbolic LCA distance to the origin~(\cref{alg:decoding}).
Intuitively, because $\dec(\cdot)$ uses LCA distances to the origin (and not pairwise distances), it can recover the underlying tree that is directly being optimized by \name{} (which is only defined through these LCA depths).
\begin{small}
\begin{algorithm}
\caption{Hyperbolic binary tree decoding $\dec(Z)$}
\label{alg:decoding}
\begin{algorithmic}[1]
        \State \textbf{Input:} Embeddings $Z=\{z_1, \ldots, z_n\}$; \textbf{Output:} Rooted binary tree with $n$ leaves.
        \State $F\leftarrow\ $forest $(\{i\}: i \in [n])$
        \State $S\leftarrow\{(i, j)$: pairs sorted by deepest hyperbolic LCA ($d_o(z_i \vee z_j)$)\};
        \For{$(i, j)\in S$}
            \If{$i$ and $j$ not in same tree in $F$}
            \State $r_i, r_j \leftarrow$ roots of trees containing $i$, $j$;
            \State create new node with children $r_i, r_j$;
            \EndIf
    \EndFor
        \State \textbf{return} $F$ (which is a binary tree at the end of the algorithm)
\end{algorithmic}
\end{algorithm}
\end{small}

\subsection{Approximation ratio result}\label{subsec:theory}
Assuming one can perfectly solve the constrained optimization in~\cref{eq:diff}, our main result is that the solution of the continuous optimization, once decoded, recovers the optimal tree with a constant approximation factor, which can be made arbitrarily small.
We first provide more conditions on the constrained set of embeddings in~\cref{eq:diff} by introducing a special set of \textit{spread} embeddings $\mathcal{Z}\subset\mathbb{B}_2^n$:
\begin{restatable}[]{definition}{spread}\label{def:spread}
    An embedding $Z \in \mathbb{B}_2^n$ is called \emph{spread} if for every triplet $(i, j, k)$:
    \begin{equation}
        \label{eq:spread}
        \max \{d_o(z_i \vee z_j), d_o(z_i \vee z_k), d_o(z_j \vee z_k) \} - \min \{d_o(z_i \vee z_j), d_o(z_i \vee z_k), d_o(z_j \vee z_k) \} >
        \delta \cdot O(n),
    \end{equation}
    where $\delta$ is a constant of hyperbolic space (Gromov's delta hyperbolicity, see~\cref{appendix:gromov}).
\end{restatable}
Intuitively, the spread constraints spread points apart and force hyperbolic LCAs to be distinguishable from each other.
Theoretically, we show that this induces a direct correspondence between embeddings and binary trees, i.e. the embeddings yield a metric that is close to that of a binary tree metric. 
More concretely, every spread leaf embedding $Z \in \mathcal{Z}$ is compatible with an underlying tree $T$ in the sense that our decoding algorithm is guaranteed to return the underlying tree.
This is the essence of the correspondence between spread embeddings and trees, which is the main ingredient behind~\cref{thm:main}.
\begin{restatable}[]{theorem}{apx}\label{thm:main}
Consider a dataset with $n$ datapoints and pairwise similarities $(w_{ij})$ and
let $T^*$ be the solution of~\cref{eq:opt}.
Let $\mathcal{Z}$ be the set of spread embeddings and $Z^*\in\mathcal{Z}$ be the solution of~\cref{eq:diff} for some $\tau>0$.
For any $\varepsilon>0$, if $\tau\le\mathcal{O}(1/\log (1 /\varepsilon))$, then:
\begin{equation}\label{eq:apx}
\begin{split}
     \frac{C_{\mathrm{Dasgupta}}(\dec(Z^*);w)}{C_{\mathrm{Dasgupta}}(T^*;w)} \le 1  + \varepsilon.
\end{split} 
\end{equation}
\end{restatable}
The proof is detailed in~\cref{appendix:main_result}. 
The insight is that when embeddings are sufficiently spread out, there is an equivalence between leaf embeddings and trees.
In one direction, we show that \textit{any} spread embedding $Z$ has a continuous cost $C_{\mathrm{\name{}}}(Z;w,\tau)$ close to the discrete cost $C_{\mathrm{Dasgupta}}(T;w)$ of some underlying tree $T$.
Conversely, we can leverage classical hyperbolic embedding results~\cite{sarkar2011low} to show that Dasgupta's optimum $T^*$ can be embedded as a spread embedding in $\mathcal{Z}$.

\paragraph{Optimization challenges} 
Note that smaller $\tau$ values in~\cref{eq:apx} lead to better approximation guarantees but make the optimization more difficult.
Indeed, hardness in the discrete problem arises from the difficulty in enumerating all possible solutions with discrete exhaustive search;
on the other hand, in this continuous setting it is easier to specify the solution, which trades off for challenges arising from nonconvex optimization---i.e. how do we find the global minimizer of the constrained optimization problem in~\cref{eq:diff}.
In our experiments, we use a small constant for $\tau$ following standard practice.
To optimize over the set of spread embeddings, one could add triplet constraints on the embeddings to the objective function. 
We empirically find that even solutions of the \textit{unconstrained} problem found with stochastic gradient descent have a good approximation ratio for the discrete objective, compared to baseline methods~(\cref{subsec:exp_quality}).

\section{\name{} Practical Considerations}\label{sec:scale}
Here, we propose two empirical techniques to reduce the runtime complexity of \name{} and improve its scalability, namely triplet sampling and a greedy decoding algorithm.
Note that while our theoretical guarantees in~\cref{thm:main} are valid under the assumptions that we use the exact decoding and compute the exact loss, we empirically validate that \name{}, combined with these two techniques, still produces a clustering of good quality in terms of Dasgupta's cost~(\cref{subsec:analysis}). 

\paragraph{Greedy decoding} 
\cref{alg:decoding} requires solving the closest pair problem, for which an almost-quadratic $n^{2-o(1)}$ lower bound holds under plausible complexity-theoretic assumptions~\cite{alman2015probabilistic,cs2018closest}.
Instead, we propose a greedy top-down decoding that runs significantly faster. 
We normalize leaves' embeddings so that they lie on a hyperbolic diameter and sort them based on their angles in the Poincar\'e disk.
We then recursively split the data into subsets using the largest angle split (\cref{fig:angle_split} in~\cref{appendix:greedy}). 
Intuitively, if points are normalized, the LCA distance to the origin is a monotonic function of the angle between points. 
Therefore, this top-down decoding acts as a proxy for the bottom-up decoding in~\cref{alg:decoding}.
In the worst case, this algorithm runs in quadratic time.  
However, if the repeated data splits are roughly evenly sized, it can be asymptotically faster, i.e., $O(n \log n)$. 
In our experiments, we observe that this algorithm is significantly faster than the exact decoding.

\paragraph{Triplet sampling}
Computing the loss term in~\cref{eq:dasgupta_cont} requires summing over all triplets which takes $O(n^3)$ time.
Instead, we generate all unique pairs of leaves and sample a third leaf uniformly at random from all the other leaves, which yields roughly $n^2$ triplets.  
Note that an important benefit of \name{} is that no matter how big the input is, \name{} can always produce a clustering by sampling fewer triplets. 
We view this as an interesting opportunity to scale to large datasets and discuss the scalability-clustering quality tradeoff of \name{} in~\cref{subsec:analysis}.
Finally, note that triplet sampling can be made parallel, unlike agglomerative algorithms, even if both take $\mathcal{O}(n^2)$ time.

\section{Experiments}
\begin{table}
    \resizebox{\textwidth}{!}{\renewcommand{\arraystretch}{0.95}
  \centering
  \begin{tabular}{llccccccc}
    \toprule
    & & \multicolumn{1}{c}{\textbf{Zoo}} & \multicolumn{1}{c}{\textbf{Iris}}
    &  \multicolumn{1}{c}{\textbf{Glass}} & \textbf{Segmentation} &  \multicolumn{1}{c}{\textbf{Spambase}} & \multicolumn{1}{c}{\textbf{Letter}} & \textbf{CIFAR-100} 
    \\
    \cmidrule(r){2-9}
     & \# Points & \multicolumn{1}{c}{101} & \multicolumn{1}{c}{150} 
     & \multicolumn{1}{c}{214} & 2310 & \multicolumn{1}{c}{4601} & 20K 
     & 50K 
     \\
    \cmidrule(r){2-9}
      & Cost & DC.$10^{-5}$ & DC.$10^{-5}$ 
      & DC.$10^{-6}$ & DC.$10^{-9}$ & DC.$10^{-10}$ & DC.$10^{-12}$  
      & DC.$10^{-13}$ \\
    \cmidrule(r){2-9}
     & Upper Bound & 3.887 & 14.12 
     & 3.959 
     & 4.839 
     & 3.495 
     & 3.031 
     & 4.408
     \\
     & Lower Bound & 2.750 & 7.709 
     & 2.750 
     & 3.258 
     & 3.025 
     & 2.244 
     & 3.990
     \\
    \midrule
    \multirow{5}{*}{Discrete} & SL & 2.897 & 8.120 
    & 3.018 & 3.705 & 3.250 & 2.625 & 4.149 \\
    & AL & 2.829 & 7.939 
    & \underline{2.906} & 3.408 & {3.159} & 2.437  & \underline{4.056} \\
    & CL & \textbf{2.802} & {7.950} 
    & 2.939 & 3.460 & 3.184 & 2.481 & 4.078 \\
    & WL & \underline{2.827} & {7.938} 
    & {2.920} & 3.434 & 3.170 & 2.453 & 4.060 \\
    \cmidrule(r){2-9} 
    & BKM & 2.861 & 8.223 
    & 2.948 & \underline{3.375} & \underline{3.127} & \textbf{2.383}
    & \underline{4.056}
    \\
    \midrule
    \multirow{2}{*}{Continuous} & UFit & 2.896 & \underline{7.916} 
    & 2.925 & 3.560 & 3.204 & - & - \\
    \cmidrule(r){2-9}
    & \name{} & \textbf{2.802} & \textbf{7.881} 
    & \textbf{2.902} & \textbf{3.341} & \textbf{3.126} & \underline{2.384} &
    \textbf{4.053} \\
    \bottomrule
    \end{tabular}}
    \vspace{5pt}
    \caption{Clustering quality results measured in discrete Dasgupta's cost (DC). 
    Best score in bold, second best underlined. Dashes indicate that the method could not scale to large datasets.}\label{table:small_exp}
\end{table}
We first describe our experimental protocol~(\cref{subsec:exp_setup}) and evaluate the \textbf{clustering quality} of \name{} in terms of discrete Dasgupta cost~(\cref{subsec:exp_quality}).
We then analyze the clustering \textbf{quality/speed tradeoffs} of \name{}~(\cref{subsec:analysis}).
Finally, we demonstrate the \textbf{flexibility} of \name{} using end-to-end training for a downstream classification task~(\cref{subsec:flexibility}).

\subsection{Experimental setup}\label{subsec:exp_setup}
We describe our experimental setup and refer to~\cref{subsec:appendix_exp_details} for more details.  
\paragraph{Datasets}
We measure the clustering quality of \name{} on six standard datasets from the UCI Machine Learning repository,\footnote{\url{https://archive.ics.uci.edu/ml/datasets.php}} as well as CIFAR-100~\cite{krizhevsky2009learning}, which exhibits a hierarchical structure (each image belongs to a fine-grained class that is itself part of a coarse superclass). 
Note that the setting studied in this work is similarity-based HC, where the input is only pairwise similarities, rather than features representing the datapoints.
For all datasets, we use the cosine similarity to compute a complete input similarity graph.

\paragraph{Baselines}
We compare \name{} to similarity-based HC methods, including competitive agglomerative clustering approaches such as single, average, complete and Ward Linkage (SL, AL, CL and WL respectively). 
We also compare to Bisecting K-Means (BKM)~\cite{moseley2017approximation}, which is a fast top-down algorithm that splits the data into two clusters at every iteration using local search.\footnote{BKM is the direct analog of Hierarchical K-Means in the context of similarity-based HC~\cite{moseley2017approximation}.}
Finally, we compare to the recent gradient-based Ultrametric Fitting (UFit) approach~\cite{chierchia2019ultrametric}.\footnote{Note that we do not directly compare to gHHC~\cite{monath2019gradient} since this method requires input features. For completeness, we include a comparison in the~\cref{appendix:ghhc_comparison}.}

\paragraph{Evaluation metrics} 
Our goal in this work is not to show an advantage on different heuristics, but rather to optimize a single well-defined search problem to the best of our abilities.
We therefore measure the clustering quality by computing the discrete Dasgupta Cost (DC).
A lower DC means a better clustering quality.
We also report upper and lower bounds for DC (defined in the~\cref{subsec:appendix_exp_details}).
For classification experiments~(\cref{subsec:flexibility}) where the goal is to predict labels, we measure the classification accuracy.

\paragraph{Training procedure}
We train \name{} for 50 epochs (of the sampled triples) and optimize embeddings with Riemannian Adam~\cite{becigneul2018riemannian}.
We set the embedding dimension to two in all experiments, and normalize embeddings during optimization as described in the greedy decoding strategy~(\cref{sec:scale}). 
We perform a hyper-parameter search over learning rate values $[1e^{-3}, 5e^{-4}, 1e^{-4}]$ and temperature values $[1e^{-1}, 5e^{-2}, 1e^{-2}]$. 
\paragraph{Implementation}
We implemented HypHC in PyTorch and make our implementation publicly available.\footnote{\url{https://github.com/HazyResearch/HypHC}} 
To optimize the \name{} loss in~\cref{eq:dasgupta_cont}, we used the open-source Riemannian optimization software geoopt~\cite{kochurov2020geoopt}.
We conducted our experiments on a single NVIDIA Tesla P100 GPU. 

\subsection{Clustering quality}\label{subsec:exp_quality}
We report the performance of \name{}---fast version with greedy decoding and triplet sampling---in~\cref{table:small_exp}, and compare to baseline methods.
On all datasets, \name{} outperforms or matches the performance of the best discrete method, and significantly improves over UFit, the only similarity-based continuous method.
This confirms our intuition that directly optimizing a continuous relaxation of Dasgupta's objective can improve clustering quality, compared to baselines that are optimized independently of the objective.

We visualize embeddings during different iterations of \name{} on the zoo dataset in~\cref{fig:scale}. 
Colors indicate ground truth flat clusters for the datapoints and we observe that these are better separated in the dendrogram as optimization progresses. 
We note that embeddings are pushed towards the boundary of the disk, where the hyperbolic distances are more ``tree-like'' and produce a better HC. 
This illustrates our intuition that the optimal embedding will be close to a tree metric embedding.

\begin{figure}
\centering
\begin{subfigure}[b]{0.44\textwidth}
\centering
    \includegraphics[width=\textwidth]{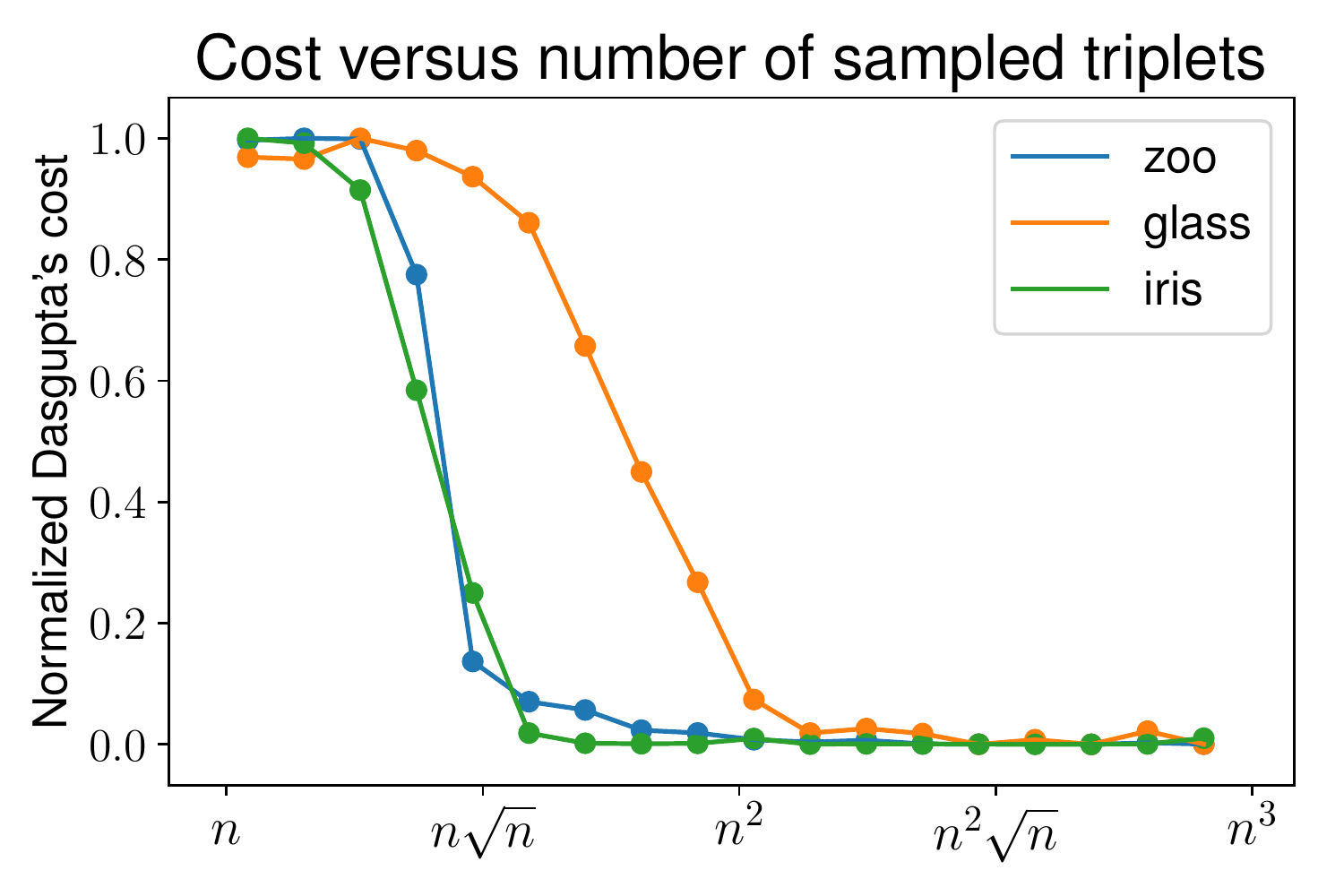}
    \caption{Normalized cost vs number of triplets.}\label{fig:triplet}
\end{subfigure}
\begin{subfigure}[b]{0.55\textwidth}
    \centering
    \resizebox{\textwidth}{!}{\renewcommand{\arraystretch}{0.95}
    \begin{tabular}{llcccc}
    \toprule
    & & \textbf{Zoo} & \textbf{Iris} &\textbf{Glass} & \textbf{Segmentation} \\
    \cmidrule(r){2-6}
    \multirow{2}{*}{Cost} & {Exact} & \textbf{2.8016}.$10^{-5}$ & \textbf{7.8783}.$10^{-5}$  &   \textbf{2.9011}.$10^{-6}$ & \textbf{3.3417}.$10^{-9}$ \\
    & {Greedy} & 
    2.8021.$10^{-5}$ & 
    7.8807.$10^{-5}$ & 
    2.9019.$10^{-6}$ &
    3.3419.$10^{-9}$ \\
    \midrule
    \multirow{3}{*}{Time (ms)} & {Exact} & 78 & 323 & 400 & 21717\\
    & {Greedy} & 3.1 & 5.2 & 6.2 & 303 \\
    \cmidrule(r){2-6}
    & {Speedup} & 25x & 62x & 64x & 72x \\
    \bottomrule
    \end{tabular}
    }
    \caption{Speed/quality analysis for the exact and greedy decoding.}\label{fig:decoding}
    \vspace{10pt}
    \resizebox{\textwidth}{!}{\renewcommand{\arraystretch}{0.95}
    \begin{tabular}{lcccc}
    \toprule
     & \textbf{Zoo} & \textbf{Iris} &\textbf{Glass} & \textbf{Segmentation} \\
     \cmidrule{2-5}
     \# Classes & 7 & 3 & 6 & 7\\
    \midrule
    LP & 41.4 & 76.7 & 46.8 & 65.3 \\
    \name{}-Two-Step & 84.8 $\pm$ 3.5 
    & 84.4 $\pm$ 1.7 
    & 50.6 $\pm$ 2.6 
    & 64.1 $\pm$ 0.9 
    \\
    \name{}-End-to-End & 
    \textbf{87.9} $\pm$ 3.8 
    & \textbf{85.6} $\pm$ 0.8 
    & \textbf{54.4} $\pm$ 2.9 
    & \textbf{67.7} $\pm$ 3.4
    \\
    \bottomrule
    \end{tabular}}
    \caption{Classification accuracy for different training strategies.}\label{table:mtl}
    \end{subfigure}
    \caption{(a): Triplet sampling analysis. (b): Downstream classification task. (c): Decoding analysis.}
\end{figure}
\subsection{Analysis}\label{subsec:analysis}
In our experiments, we used the greedy decoding Algorithm and triplet sampling with $\mathcal{O}(n^2)$ triplets~(\cref{sec:scale}). 
Since our theoretical guarantees apply to the full triplet loss and exact decoding, we are interested in understanding how much quality is lost by using these two empirical techniques. 

\paragraph{Decoding} 
We report HC costs obtained using the exact~(\cref{subsec:dec}) and the greedy~(\cref{sec:scale}) decoding in~\cref{fig:decoding}, as well as the corresponding runtime in milliseconds (ms), averaged over 10 runs. 
We observe that the greedy decoding is approximately 60 times faster than the exact decoding, while still achieving almost the same cost.
This confirms that, when embeddings are normalized in two dimensions, using angles as a proxy for the LCA's distances to the origin is a valid approach. 

\paragraph{Triplet sampling}
We plot the discrete cost (normalized to be in $[0, 1]$) of \name{} for different number of sampled triplets in~\cref{fig:triplet}.
As expected, we note that increasing the number of sampled triplets reduces DC, and therefore improves the clustering quality.
This tradeoff forms one of the key benefits of using gradient-based methods: \name{} can always produce a hierarchical clustering, no matter how large the input is, at the cost of reducing clustering quality or increasing runtime. 
In our experiments, we find that $O(n^2)$ is sufficient to achieve good clustering quality, which is the most expensive step of \name{} in terms of runtime complexity.
Future work could explore better triplet sampling strategies, to potentially use less triplets while still achieving a good clustering quality. 

\subsection{End-to-end training}\label{subsec:flexibility}
Here, we demonstrate the flexibility of our approach and the benefits of joint training on a downstream similarity-based classification task.
Since \name{} is optimized with gradient-descent, it can be used in conjunction with any standard ML pipeline, such as downstream classification.
We consider four of the HC datasets that come with categorical labels for leaf nodes, split into training, testing and validation sets (30/60/10\% splits).
We follow a standard graph-based semi-supervised learning setting, where all the nodes (trainining/validation/testing) are available at training time, but only training labels can be used to train the models.
We use the embeddings learned by \name{} as input features for a hyperbolic logistic regression model~\cite{ganea2018hyperbolic}.
Note that none of the other discrete HC methods apply here, since these do not produce representations.
In~\cref{table:mtl}, we compare jointly optimizing the \name{} and the classification loss (End-to-End) versus a two-step embed-then-classify approach which trains the classification module using freezed \name{} embeddings (Two-Step) (average scores and standard deviation computed over 5 runs).
We also compare to Label Propagation (LP), which is a simple graph-based semi-supervised algorithm that does not perform any clustering step. 
We find that LP is outperformed by both the end-to-end and the two-step approaches on most datasets, suggesting that clustering learns meaningful partitions of the input similarity graph.
Further, we observe that end-to-end training improves classification accuracy by up to 3.8\%, confirming the benefits of a differentiable method for HC.

\section{Conclusion}
We introduced \name{}, a differentiable approach to learn HC with gradient-descent in hyperbolic space.
\name{} uses a novel technical approach to optimize over discrete trees, by showing an equivalence between trees and constrained hyperbolic embeddings.
Theoretically, \name{} has a $(1+\varepsilon)$-factor approximation to the minimizer of Dasgupta's cost, and empirically, \name{} outperforms existing HC algorithms. 
While our theoretical analysis assumes a perfect optimization, interesting future directions include a better characterization of the hardness arising from optimization challenges, as well as providing an approximation for the continuous optimum.
While no constant factor approximation of the continuous optimum is possible, achieving better (e.g. polylogarithmic) approximations, is an interesting future direction for this work. 
Finally, we note that our continuous optimization framework can be extended beyond the scope of HC, to applications that require searching of discrete tree structures, such as constituency parsing in natural language processing. 

\section*{Broader Impact}
Clustering is arguably one of the most commonly used tools in computer science applications. 
Here, we study a variation where the goal is to output a hierarchy over clusters, as data often contain hierarchical structures. 
We believe our approach based on triplet sampling and optimization should not raise any ethical considerations, to the extent that the input data for our algorithm is unbiased. 
Of course, bias in data is by itself another challenging problem, as biases can lead to unfair clustering and discriminatory decisions for different datapoints. 
However here we study a downstream application, \emph{after} data has been collected. 
As such, we hope only a positive impact can emerge from our work, by more faithfully finding hierarchies in biological, financial, or network data, as these are only some of the applications that we listed in the introduction.

\section*{Acknowledgments}
We gratefully acknowledge the support of DARPA under Nos. FA86501827865 (SDH) and FA86501827882 (ASED); NIH under No. U54EB020405 (Mobilize), NSF under Nos. CCF1763315 (Beyond Sparsity), CCF1563078 (Volume to Velocity), and 1937301 (RTML); ONR under No. N000141712266 (Unifying Weak Supervision); the Moore Foundation, NXP, Xilinx, LETI-CEA, Intel, IBM, Microsoft, NEC, Toshiba, TSMC, ARM, Hitachi, BASF, Accenture, Ericsson, Qualcomm, Analog Devices, the Okawa Foundation, American Family Insurance, Google Cloud, Swiss Re, the HAI-AWS Cloud Credits for Research program, TOTAL, and members of the Stanford DAWN project: Teradata, Facebook, Google, Ant Financial, NEC, VMWare, and Infosys. The U.S. Government is authorized to reproduce and distribute reprints for Governmental purposes notwithstanding any copyright notation thereon. Any opinions, findings, and conclusions or recommendations expressed in this material are those of the authors and do not necessarily reflect the views, policies, or endorsements, either expressed or implied, of DARPA, NIH, ONR, or the U.S. Government.

\small
\bibliographystyle{plain}
\bibliography{ref}

\newpage
\normalsize
\appendix


\section{Preliminaries}\label{appendix:preliminaries}
We first introduce our notational conventions in~\cref{appendix:notations} and define the notion of hyperbolicity in~\cref{appendix:gromov}.
In~\cref{appendix:tree_like}, we review some useful results of hyperbolic metrics that will be used throughout our proofs.
\subsection{Notations}\label{appendix:notations}
\paragraph{Binary trees}
A \emph{binary tree} is one that has all degrees either $1$ (a leaf node) or $3$ (internal node).\footnote{We use undirected trees in our proofs.}
A binary tree with $n$ leaves has exactly $n-1$ internal nodes (of degree $3$).
A \textit{rooted binary tree} is a binary tree such that one of its internal nodes (i.e. non-leaf node), the root, has degree exactly $2$.
Note that if we think of one leaf in a binary tree as the root, then removing it (and letting its unique neighbor be the new root) converts this into a \textit{rooted binary tree} in the traditional sense (See~\cref{fig:unrooted_tree} and~\cref{fig:rooted_tree} for examples of rooted and unrooted binary trees).
By convention, we let $T$ denote any binary tree with leaves $1, \dots, n$ and root $0$ if its rooted.

\paragraph{Tree and hyperbolic metrics}
With $i, j, k$, we refer to nodes in  $T$, and $z_i, z_j, z_k$ correspond to points in a hyperbolic embedding $Z=\{z_1, \dots, z_n\}\subset\mathbb{B}_2^n$.
With $o$ or $z_0$, we denote the origin of hyperbolic space and $d_{\mathbb{B}_2}$ is the hyperbolic metric; every tree $T$ defines a tree metric $d_T$. We overload $d(i, j) := d_T(i, j)$ and $d(z_i, z_j) := d_{\mathbb{B}_2}(z_i, z_j)$ when the types are clear.
There is a close correspondence between hyperbolic and tree metrics, and we define the notion of quasi-isometries, which we will use throughout our proofs.
\begin{definition}[Quasi-isometric embedding]\label{defn:quasi}
    \label{def:quasi}
    Let $T$ be a binary tree on $(n+1)$ leaves and $Z=\{z_0,\ldots,z_n\}$ an embedding set.
    The pair $(Z, T)$ is $(1+\varepsilon, \kappa)$-quasi-isometric if:
    \begin{align*}
        \label{eq:quasi}
        d(i, j) \le d(z_i, z_j) \le (1+\varepsilon) d(i, j) + \kappa,
    \end{align*}
     for all $0 \le i, j \le n$.%
\end{definition}
\begin{remark}
Note that \cref{def:quasi} is equivalent to saying there is a quasi-isometric embedding from $\{z_0, \dots, z_n\}$ to the root and leaves of $T$, i.e. the two metrics agree up to a linear transform.
Compared to the usual definition of quasi-isometry, we consider only a one-sided version for convenience in the proofs.
\end{remark}

\paragraph{Node depth} We overload $d_0(i) = d_T(0, i)$ to refer to the distance from node $i$ to the root in $T$, and $d_o(z_i) = d_{\mathbb{B}_2}(z_0, z_i)$ to be the distance to the origin in $\mathbb{B}_2$.
Intuitively, $d_0(\cdot)$ and $d_o(\cdot)$ represent the ``depth'' of a node or point.
Note that the depths are dependent on the choice of a base point as the root of a tree or origin of the space. 
We always use $o$ for $\mathbb{B}_2$ and the node indexed by $0$ for trees.

\paragraph{LCA}
We let $i\vee j$ denote the LCA of two leaf nodes in $T$, and $z_i\vee z_j$ denote the hyperbolic LCA defined in~\cref{eq:hyp_lca}.
In particular, we say that $\{i, j|k\}_T$ holds if the LCA of $(i, j)$ has a larger depth than that of $(i, k)$ and $(j, k)$, i.e. $d_0(i\vee j)\ge\mathrm{max}\{d_0(i\vee k), d_0(j\vee k)\}$. 
Similarly, we say that $\{z_i, z_j|z_k\}_{\mathbb{B}_2}$ holds if $d_o(z_i\vee z_j)\ge\mathrm{max}\{d_o(z_i\vee z_k), d_o(z_j\vee z_k)\}$.
We overload $\{i, j|k\}_T=\{i, j|k\}$ and $\{z_i, z_j|z_k\}_{\mathbb{B}_2}=\{z_i, z_j|z_k\}$ when the types are clear.

\subsection{Gromov's delta hyperbolicity}\label{appendix:gromov}
We define the Gromov~\cite{gromov1987hyperbolic} product which can be used to define $\delta$-hyperbolic spaces.
\begin{definition}[Gromov product]%
    \label{def:gromov}
    In any metric space $(X, d)$, the Gromov product of points $x, y\in X$ with respect to a third point $z\in X$ is:
    \[
        \langle x, y \rangle_z = \frac{1}{2}\left( d(x, z) + d(z, y) - d(x, y) \right).
    \]
\end{definition}
When the base point $z$ is taken to be the origin of $\mathbb{B}_2$ or the root of a tree, we shorten this to $\langle x, y \rangle$ unambiguously.

A key characterization of hyperbolic spaces is the notion of $\delta$-hyperbolicity.
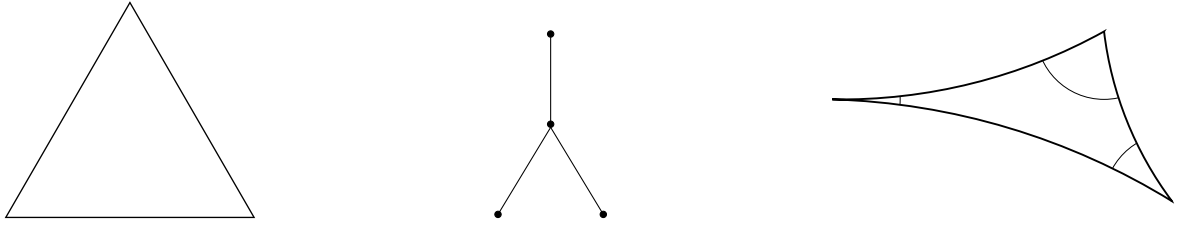
\begin{figure}[t]
    \begin{subfigure}[b]{0.3\textwidth}
    \begin{center}
    \resizebox{0.7\textwidth}{!}{\renewcommand{\arraystretch}{0.95}
    \begin{tikzpicture}
    \node[draw, minimum size=3cm,regular polygon,regular polygon sides=3] (a) {};
    \end{tikzpicture}}    
    \end{center}
    \caption{Euclidean triangle (not $\delta$-slim).}
    \end{subfigure}
    \hfill
    \begin{subfigure}[b]{0.3\textwidth}
    \begin{center}
        \begin{istgame}[font=\footnotesize]
    \xtdistance{12mm}{16mm}
    \istroot(0){}
      \istb{}[l]
      \endist
    \xtdistance{12mm}{14mm}
    \istroot(1)(0-1)<135>{}
      \istb{}[l]
      \istb{}[r]
      \endist
    \xtdistance{12mm}{14mm}
    \istroot(2)(1-1)<45>{}
      \endist
    \xtdistance{12mm}{14mm}
    \istroot(3)(1-2)<45>{}
     \endist
    \end{istgame}
    \end{center}
    \caption{$0$-slim triangle in a tree.}
    \end{subfigure}
    \hfill
    \begin{subfigure}[b]{0.3\textwidth}
    \begin{center}
    \resizebox{\textwidth}{!}{\renewcommand{\arraystretch}{0.95}
        \begin{tikzpicture}[font=\sffamily]
     \path (0,0) coordinate (A) (4,1) coordinate (B) (5,-1.5) coordinate (C);
     \draw[thick,path picture={
     \foreach \X in {A,B,C}
     {\draw[line width=0.4pt] (\X) circle (1);}}]
     let \p1=($(B)-(A)$),\p2=($(C)-(B)$),\p3=($(C)-(A)$),
     \n1={atan2(\y1,\x1)},\n2={atan2(\y2,\x2)},\n3={atan2(\y3,\x3)},
     \n4={veclen(\y1,\x1)},\n5={veclen(\y2,\x2)},\n6={veclen(\y3,\x3)} in
     (A) node[left]{}  arc(-90-15+\n1:-90+15+\n1:{\n4/(2*sin(15))})
     --(B) node[above right]{} 
      arc(-90-15+\n2:-90+15+\n2:{\n5/(2*sin(15))})
     --(C) node[below]{} 
     arc(90-15+\n3:90+15+\n3:{\n6/(2*sin(15))}) -- cycle;
     \node at (barycentric cs:A=1,B=1,C=1) {};
    \end{tikzpicture}}
    \end{center}
    \caption{$\delta$-slim triangle.}
    \end{subfigure}
    \caption{Illustration of the notion of $\delta$-slim triangles.}
    \label{fig:delta_thin}
\end{figure}
\begin{definition}[$\delta$-hyperbolicity, four-point condition]%
    \label{def:hyperbolicity}
    A metric space $(X, d)$ is $\delta$-hyperbolic if there exists $\delta\ge 0$ such that for all $w, x, y, z$ in $X$:
\begin{align*}
    \langle x, y\rangle_z\ge\mathrm{min}\{\langle x, w\rangle_z, \langle w, y\rangle_z\}-\delta.
\end{align*}
\end{definition}
\begin{example}
    The hyperbolic space $\mathbb{B}_2$ is $\log 3$-hyperbolic.
\end{example}
\begin{example}
    Metric trees are $0$-hyperbolic. 
\end{example}
\begin{example}
    The Euclidean space $\mathbb{R}^n$ is not $\delta$-hyperbolic. 
\end{example}
Up to changing $\delta$ by a constant multiple, there are many equivalent variations of the notion of $\delta$-hyperbolicity.
In particular, one intuitive interpretation of $\delta$-hyperbolic spaces is using the notion of $\delta$-slim triangles, which says that any triangle in a $\delta$-hyperbolic space has distance from any side to the other two less than $\delta$.
This is not true in Euclidean space, for instance the midpoint of a large isosceles triangle might be far from the other two sides~(\cref{fig:delta_thin}). 

\subsection{Tree-likeness of the hyperbolic space}\label{appendix:tree_like}
An important result in the theory of hyperbolic metric spaces is their tree-likeness.
We review two useful results of $\delta$-hyperbolic metrics. 
First, our notion of LCA depth is closely related to the Gromov product, which is exactly the tree depth for $0$-hyperbolic metrics~(\cref{lmm:lca-gromov}).
Second, any finite set of points in $\mathbb{B}_2$ can be embedded in a binary tree~(\cref{prop:treelike}):
\begin{lemma}[LCA depth is close to the Gromov product]%
    \label{lmm:lca-gromov}
    For any $i, j \in T$,
    \begin{align*}
        d_0(i \vee j) = \langle i, j \rangle.
    \end{align*}
    There exists $\delta>0$ such that for any $z_i, z_j \in \mathbb{B}_2$,
    \begin{align*}
        \langle z_i, z_j \rangle \le d_o(z_i \vee z_j) \le \langle z_i, z_j \rangle + \delta.
    \end{align*}
\end{lemma}
\begin{proof} 
This is a direct application of Lemma 6.1 and 6.2 in~\cite{bowditch2007course} to tree metrics (which are $0$-hyperbolic) and to the hyperbolic space $\mathbb{B}_2$ (which is $\delta$-hyperbolic).\footnote{Note that the notion of $\delta$-hyperbolicity in~\cite{bowditch2007course} uses the incenter condition, which is equivalent to the four point condition~(\cref{def:hyperbolicity}), up to changing $\delta$ by a constant multiple. In standard hyperbolic space,~\cref{lmm:lca-gromov} is in fact satisfied for $\delta=\log 3$.}
\end{proof}
\begin{proposition}[Tree-likeness of hyperbolic space]%
    \label{prop:treelike}
    There is a constant $C_n$ such that for any set of points $\{z_0, z_1, \dots, z_n\} \subset \mathbb{B}_2^{n+1}$,
    there is a binary tree $T$ on leaves $0, 1, \dots, n$ such that:
    \begin{equation}
        \label{eq:treelike}
        \forall 0 \le i, j \le n: d_T(i, j) \le d_{\mathbb{B}_2}(z_i, z_j) \le d_T(i, j) + C_n,
    \end{equation}
    with $C_n = \delta \cdot O(n)$.
\end{proposition}
\begin{proof}%
    The statement of \cref{prop:treelike} without the binary or leaf condition is a standard result (See Proposition 6.7 in~\cite{bowditch2007course}).
    That is, there exists a tree $T$ with $n$ nodes (not necessarily binary) satisfying~\cref{eq:treelike}.
    We modify $T$ to satisfy the leaf constraint (i.e. $\{z_0, z_1,\ldots,z_n\}$ are leaves' embeddings) and to satisfy the binary condition (i.e. $T$ is binary in the sense that every node has degree $1$ or $3$).
    \paragraph{Leaf condition}
    Let $k\in[n]$ be a node in $T$ and $m$ be the minimum edge length in $T$. If $k$ has degree greater than $1$ (i.e. $k$ is not a leaf node), then we shrink every edge connected to it by some constant $c < \mathrm{min}\{\delta, \frac{m}{n}\}$, 
    create a dummy node $p$ in place of $k$ and connect $k$ to it~(\cref{fig:leaf_cond}).
    Now all tree distances involving $k$ are the same, and all other distances going through $p$ are shrunk by at most $2c$.
    The resulting tree $T'$ has $n$ leaves and is such that: 
    \begin{align*}
        \forall 0 \le i, j \le n: d_{T'}(i, j) \le d_{T}(i, j) \le d_{T'}(i, j) + 2 c.
    \end{align*}
    \paragraph{Binary condition}
    Next, we modify $T'$ to be binary.
    For every node of degree $2$, simply delete it, which does not affect the tree metric restricted to $[n]$ (we could not have deleted a node in $[n]$ since they are all leaves now).
    For every node of degree $4$ or more, we replace it by multiple copies connected by edges of length $c$ and decrease original edges by $c$ every time we create a copy of the original node~(\cref{fig:binary_cond}), which causes the distances to shrink by at most $2c(n-1)$ (case of star trees). 
    Therefore, distances in this new binary tree $T''$ satisfy: 
    \begin{align*}
        \forall 0 \le i, j \le n: d_{T''}(i, j) \le d_{T'}(i, j) \le d_{T''}(i, j) + 2 c(n-1).
    \end{align*}
    With this final binary tree on $(n+1)$ leaves $T''$, we have:
    \begin{align*}
        \forall 0 \le i, j \le n: d_{T''}(i, j)\le d_\mathbb{B}(z_i, z_j) \le d_{T''}(i, j) + C_n + 2c + 2 c(n-1).
    \end{align*}
    Thus the statement holds for $C_n' = C_n + 2c n$, which is still $\delta \cdot O(n)$ since $c<\delta$.
\end{proof}
\begin{figure}[t]
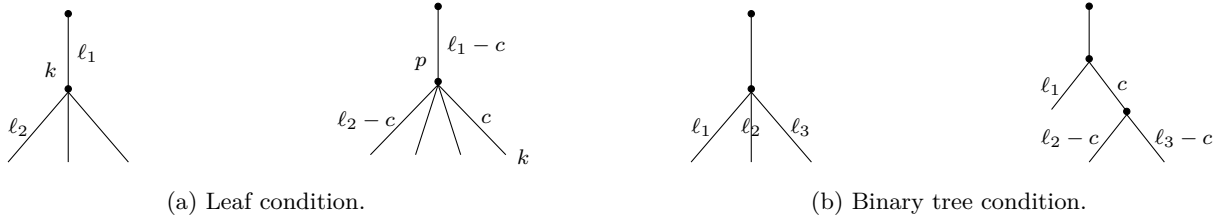

    \centering
    \begin{subfigure}[b]{0.45\textwidth}
    \input{figures/tree_leaf_cond}
    \caption{Leaf condition.}\label{fig:leaf_cond}
    \end{subfigure}
    \hfill
    \centering
    \begin{subfigure}[b]{0.45\textwidth}
    \input{figures/tree_binary_cond}
    \caption{Binary tree condition.}\label{fig:binary_cond}
    \end{subfigure}
    \caption{Tree transformations used to satisfy the leaf and the binary conditions in~\cref{prop:treelike}}
\end{figure}
\begin{remark}
Using the definition of quasi-isometries~(\cref{def:quasi}),~\cref{prop:treelike} implies that for any embedding $Z\in\mathbb{B}_2^n$, there exists a binary tree $T$ on $n$ leaves such that there is a $(1, C_n)$-quasi-isometry from the leaves of $T$ to $Z$.
\end{remark}

\section{Hyperbolic LCA construction}\label{appendix:lca_construction}
We detail the calculations used to compute the hyperbolic LCA and its distance to the origin.
\lcanorm*
\begin{proof}
We use circle inversions to show this result. Circle inversions are Euclidean transformations on the plane that map circles to circles and preserve angles~\cite{brannan2011geometry}, and can represent hyperbolic reflections (isometric transformations) along hyperbolic geodesics. 
In particular, the circle inversion formula can be used to recover the center of the circle that is orthogonal to the boundary of the disk and that coincides with the geodesic between two given points. 
Consider the circle defined by the geodesic connecting $x$, and $y$ and let $R$ denote its radius, $p$ the orthogonal projection of $o$ onto this circle and $\Delta=\overline{oo'}$ denote the distance between the origin and the circle center~(\cref{fig:lca_norm}). 
By the circle inversion property and using the fact that the Poincar\'e disk has radius one, we have $1=(\Delta - R)(\Delta + R)$, which yields:
\begin{equation}
    R^2=\Delta^2-1.\label{eq:inversion_formula}
\end{equation}
Additionally, if $\theta=\angle xoy$ and $\alpha=\angle pox$, we have by the Pythagorean theorem:
\begin{equation}
    \begin{cases}
        R^2=(\Delta-||x||_2\mathrm{cos}(\alpha))^2+||x||^2_2\mathrm{sin}^2(\alpha)\\
        R^2=(\Delta-||y||_2\mathrm{cos}(\theta-\alpha))^2+||y||^2_2\mathrm{sin}^2(\theta-\alpha).
    \end{cases}
\end{equation}
This leads to the system of equations:
\begin{equation}
    \begin{cases}
        2\sqrt{R^2+1}\mathrm{cos}(\alpha)=\frac{||x||_2^2+1}{||x||_2}\\
        2\sqrt{R^2+1}\mathrm{cos}(\theta-\alpha)=\frac{||y||_2^2+1}{||y||_2}
    \end{cases}
\end{equation}
which is solved for $R$ and $\alpha$ defined in~\cref{eq:lca_norm}.
Now using~\cref{eq:inversion_formula}, we have: $||p||_2=\Delta-R=\sqrt{R^2+1}-R$.
Finally, we get the result using the hyperbolic distance function and noting that the orthogonal projection of a point on a geodesic that does not contain that point is minimizing the distance between the point and the geodesic, that is $p=x\vee y$. 
\end{proof}
\begin{figure}
    \centering
     \includegraphics[width=0.35\textwidth]{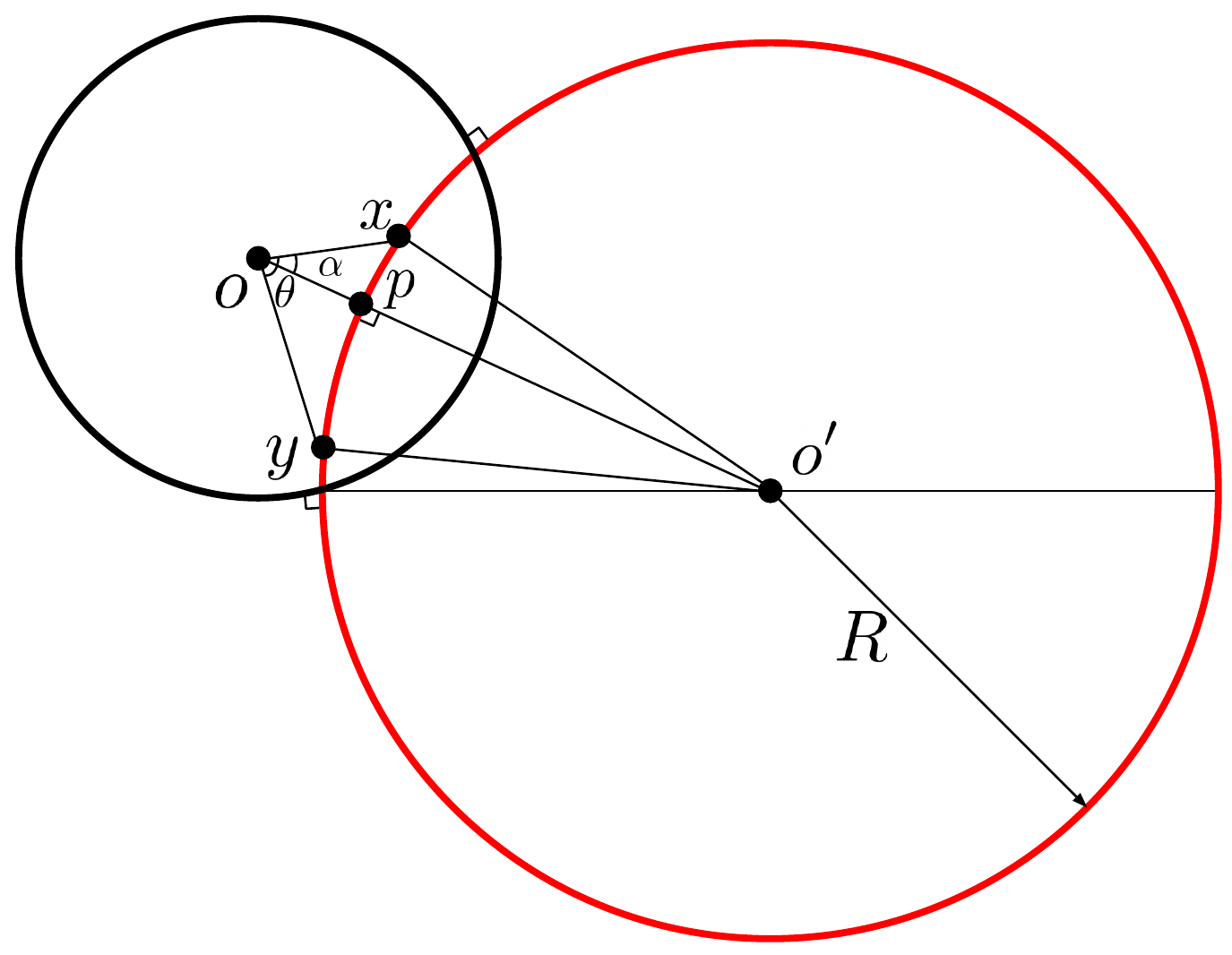}
        \caption{Circle inversion used for hyperbolic LCA construction in~\cref{lemma:lca_norm}.}\label{fig:lca_norm}
\end{figure}

\section{Proof of~\cref{thm:main}}\label{appendix:main_result}
Our goal in this section is to show our main result,~\cref{thm:main}, which gives a $(1+\varepsilon)$-approximation ratio for Dasgupta's discrete objective. 
We first give an overview of the strategy used to show~\cref{thm:main} and state the two results used in the proof~(\cref{lmm:spread-equivalence} and~\cref{lmm:encode-decode}) in~\cref{appendix:proof_outline}.
We then go over the details of~\cref{lmm:spread-equivalence} in~\cref{appendix:decoding_error} and~\cref{lmm:encode-decode} in~\cref{sec:theory-encoding}. 
\setcounter{subsection}{-1}
\subsection{Proof outline}\label{appendix:proof_outline}
Our proof of~\cref{thm:main} relies on two main results; the existence of a constrained embedding, the spread embedding set~(\cref{def:spread}), such that the rooted binary tree decoded from any embedding in that set has a discrete cost close to the continuous cost~(\cref{lmm:spread-equivalence}), and reciprocally, any rooted binary tree has a corresponding embedding in that set~(\cref{lmm:encode-decode}).

First, we define the constrained spread embedding set and give more precise conditions on the spread constants, which we'll use throughout the proof. 
\begin{definition}[Spread Embeddings]\label{def:spread-detailed}
    An embedding $Z \in \mathbb{B}_2^n$ is called \emph{spread} if for every triplet $(i, j, k)$:
    \begin{equation}
        \max \{d_o(z_i \vee z_j), d_o(z_i \vee z_k), d_o(z_j \vee z_k) \} - \min \{d_o(z_i \vee z_j), d_o(z_i \vee z_k), d_o(z_j \vee z_k) \} >
        3C_n + 2\delta + 1,
    \end{equation}
    where $\delta$ is Gromov's delta hyperbolicity~(\cref{lmm:lca-gromov}) and $C_n=\delta\cdot O(n)$ is defined in~\cref{prop:treelike}.
\end{definition}

The spread constraints force LCAs to be distinguishable from each other, which intuitively encourages binary tree metrics.
Now using this definition, we show that any spread embedding decoded using~\cref{alg:decoding} returns a tree that has a discrete cost close to the embeddings' continuous cost.
\begin{restatable}[]{lemma}{decodelemma}\label{lmm:spread-equivalence}
Let $Z\in\mathcal{Z}\subset\mathbb{B}_2^n$ be a spread embedding. Then:
\[
|C_{\mathrm{Dasgupta}}(\dec(Z);w) - C_{\mathrm{\name{}}}(Z;w,\tau)|\le 4e^{-1/\tau}\sum_{ijk}\mathrm{max}\{|w_{ij|}, |w_{ik}|, |w_{jk}|\}
\]
\end{restatable}

We then show that any rooted binary tree has a corresponding spread embedding that decodes to it.
\begin{restatable}[]{lemma}{encodelemma}\label{lmm:encode-decode}
For any unit-weight rooted binary tree $T$ on $n$ leaves, there exists a spread embedding $Z\in\mathcal{Z}\subset\mathbb{B}_2^n$ such that
$\dec(Z) = T$.
\end{restatable}
These two Lemmas show the tight equivalence between our continuous $C_{\mathrm{\name{}}}$ cost on spread embeddings, and the discrete Dasgupta cost $C_{\mathrm{Dasgupta}}$ on rooted binary trees.
Finally, putting these together, we show our main result which is that the discrete tree returned by \name{} has a $(1+\varepsilon)$-approximation factor for Dasgupta's minimum~(\cref{thm:main}). 
\apx*
\begin{proof}
Let $\mathcal{Z}\subset\mathbb{B}_2^n$ be the set of spread leaves' embeddings embeddings~(\cref{def:spread-detailed}), $\tau>0$ and:
\begin{align*}
    T^*&={\mathrm{argmin}}_T\ C_{\mathrm{Dasgupta}}(T;w)\\
    Z^*&={\mathrm{argmin}}_{Z\in\mathcal{Z}}\ C_{\name{}}(Z; w,\tau).
\end{align*}
WLOG, assume that all edges in $T^*$ have unit weight.\footnote{Any weighted version of $T^*$ would achieve the same Dasgupta cost, so we consider the unit weight case for simplicity.}
Since $T^*$ is a unit-weight rooted binary, we can apply~\cref{lmm:encode-decode} to find $Z\in\mathcal{Z}$ such that $\dec(Z)=T^*$. 

Next, let $\Delta\coloneqq C_{\mathrm{Dasgupta}}(\dec(Z^*);w)-C_{\mathrm{Dasgupta}}(T^*;w)$.
We have:
\begin{align*}
    \begin{split}
      0\le \Delta 
      &\le C_{\mathrm{Dasgupta}}(\dec(Z^*);w) - C_{\mathrm{\name{}}}(Z^*;w,\tau) + C_{\mathrm{\name{}}}(Z;w,\tau) - C_{\mathrm{Dasgupta}}(T^*;w)\\
      &= C_{\mathrm{Dasgupta}}(\dec(Z^*);w) - C_{\mathrm{\name{}}}(Z^*;w,\tau) + C_{\mathrm{\name{}}}(Z;w,\tau) - C_{\mathrm{Dasgupta}}(\dec(Z);w)\\
      &\le 2\ \mathrm{sup}_{Z\in\mathcal{Z}}|C_{\mathrm{Dasgupta}}(\dec(Z);w) - C_{\mathrm{\name{}}}(Z;w,\tau)|\\
      &\le 8\ e^{-1/\tau}\sum_{ijk}\mathrm{max}\{|w_{ij}|, |w_{ik}|, |w_{jk}|\}.
    \end{split}
\end{align*}
The first inequality follow from the fact that $C_{\mathrm{\name{}}}(Z^*;w,\tau) \le C_{\mathrm{\name{}}}(Z;w,\tau)$ since $Z^*$ is the minimizer, and the last inequality uses~\cref{lmm:spread-equivalence}.
Then:
\begin{equation}
    \frac{C_{\mathrm{Dasgupta}}(\dec(Z^*);w)}{C_{\mathrm{Dasgupta}}(T^*;w)}\le 1 + 8\ e^{-1/\tau}\bigg(\frac{\sum_{ijk}\mathrm{max}\{w_{ij}, w_{ik}, w_{jk}\}}{ C_{\mathrm{Dasgupta}}(T^*;w)}\bigg).
\end{equation}
Since $C_{\mathrm{Dasgupta}}(T^*;w)\ge\sum_{ijk}\mathrm{min}\{w_{ij}+w_{ik}, w_{ij}+w_{jk}, w_{ik}+w_{jk}\}+2 \sum_{ij}w_{ij}$, we finally have:
\begin{equation}
    \frac{C_{\mathrm{Dasgupta}}(\dec(Z^*);w)}{C_{\mathrm{Dasgupta}}(T^*;w)}\le 1 + 
    \mathcal{O}(e^{-1/\tau}).
\end{equation}
\end{proof}
\begin{remark}
    At a high level, this bound suggests that a better approximation of the argmax function (lower $\tau$) gives a better approximation for Dasgupta's discrete objective.
\end{remark}
\begin{remark}%
    The optimization set $\mathcal{Z}$ defines a local constraint on every triplet $z_i, z_j, z_k$ (\cref{def:spread}),
    and can be enforced on triplets concurrently with sampling them for the main loss~(\cref{eq:dasgupta_cont}).
    The max and min operations can be relaxed using softmax, and the separation
    can be enforced with any auxiliary constraint (e.g.\ hinge loss) .
\end{remark}

\subsection{Proof of~\cref{lmm:spread-equivalence}}\label{appendix:decoding_error}
The goal of this section is to show~\cref{lmm:spread-equivalence}.
We first introduce a notion of LCA agreement~(\cref{subsec:lca_agreement}).
We show any spread embedding is in LCA agreement with some unrooted binary tree $T$, and gets decoded into the rooted version of $T$~(\cref{lemma:decoding} in~\cref{subsec:lca_agreement_condition}).
This then allows us to show~\cref{lmm:spread-equivalence} in~\cref{subsubsec:proof_decode}.
\subsubsection{LCA Agreement}\label{subsec:lca_agreement}
We seek to understand when the continuous \name{} cost of an embedding is close to the discrete cost of the tree decoded with~\cref{alg:decoding}.
The discrete and continuous costs both depend on the ordering of triplets, in terms of their LCAs' depths (i.e. distances to the origin or root). 
If this ordering is the same in the continuous embedding space and in the decoded discrete tree, then the costs should be close, up the the continuous approximation error of the softmax function. 
This motivates us to introduce the notion of LCA agreement. 
\begin{definition}
Given a metric space $(X, d)$, and binary operation $\vee : X \times X \to X$, we define the set of \emph{LCA triplets} on any points $x_1, \dots, x_n \in X$ with respect to the base point $x_0$ to be
\[
\{ \{i, j | k\}, i, j, k \in [n] :  d(x_0, x_i \vee x_j) > \min \{ d(x_0, x_i \vee x_k), d(x_0, x_j \vee x_k) \}.
\]
\end{definition}
\begin{definition}[LCA agreement]\label{defn:lca_agreement}
We say two sets of points $x_1, \dots, x_n \in X^n$ and $y_1, \dots, y_n \in Y^n$ (possibly in different metric spaces) are \emph{in LCA agreement} or \emph{LCA equivalent} with respect to $x_0$ and $y_0$, if their set of LCA triplets are identical.
\end{definition}
Note that LCA agreement is a relative notion and depends on the base point. 
In our proofs, we always refer to LCA agreement as being with respect to the origin $z_0=o\in\mathbb{B}_2$ for embeddings, and with respect to the node labelled $0$ for trees. 

Intuitively, LCA agreement implies that the LCAs of an embedding are consistent with a ground truth tree. 
Since our proposed decoding algorithm only relies on LCA distances, we can show that if there exist a unrooted binary tree $T$ such that $(Z, T)$ are in LCA agreement, then $\dec(Z)$ recovers the rooted version of $T$. 
\begin{lemma}\label{lemma:decoding}
    Let $Z$ be an embedding and $T$ be a binary tree on $(n+1)$ leaves (not rooted) such that $(Z\cup\{z_0=o\}, T)$ are in LCA agreement. 
    Let $T'$ be the rooted binary tree on $n$ leaves that is obtained by removing the leaf $0$ from $T$.
    Then $\dec(Z)=T'$.
\end{lemma}
\begin{proof}
To be a little more formal, define $\dec(\cdot)$ to depend on three things: a set of points ($Z$ or $T$), a base point ($z_0=o$ or $0$), and a LCA construction ($\vee_{\mathbb{B}_2}$ or $\vee_T$). 

Let $T'$ be the rooted binary tree that is obtained by removing the leaf $0$ from $T$, and relabelling its neighbor $r$ in $T$ as $0$, the root of $T'$~(see~\cref{fig:unrooted_tree} and~~\cref{fig:rooted_tree}).
Because $Z$ and $T$ are in LCA agreement, the tree constructed from $\dec(Z; o, \vee_{\mathbb{B}_2})$ is exactly the same as the tree constructed from $\dec(T; 0, \vee_T)$.
Thus we only have to show that $T'= \dec(T; 0, \vee_T)$.

We use an induction argument to show that  $T'= \dec(T; 0, \vee_T)$.
Let $r$ be the unique neighbor of $0$ in $T$, and let $T_0$ and $T_1$ be the two trees on the other edges of $r$ aside from the one pointing to $0$~(\cref{fig:decoding_proof}).
Note that any distances $d_T(0, i \vee j)$ for any $i, j$ both contained in $T_0$ or $T_1$ are strictly smaller than distances $d_T(0, i \vee j)$ for any $i \in T_0, j \in T_1$.
Therefore the decoding algorithm merges all pairs $(i, j) \in T_0$ and $(i, j) \in T_1$ first.
Inductively, this decoding algorithm will exactly return $T_0$ and $T_1$.

Finally, the decoding algorithm will merge any two $i \in T_0, j \in T_1$. This creates a new parent and connects them to the roots of $T_0, T_1$ (Algorithm 1, line 7).
The structure of this tree is therefore the same as the tree rooted at $r$, in other wise $T'$, as desired.
\end{proof}
\begin{figure}[t]
    \centering
    \begin{subfigure}[b]{0.19\textwidth}
         \centering
         \includegraphics[width=0.9\textwidth]{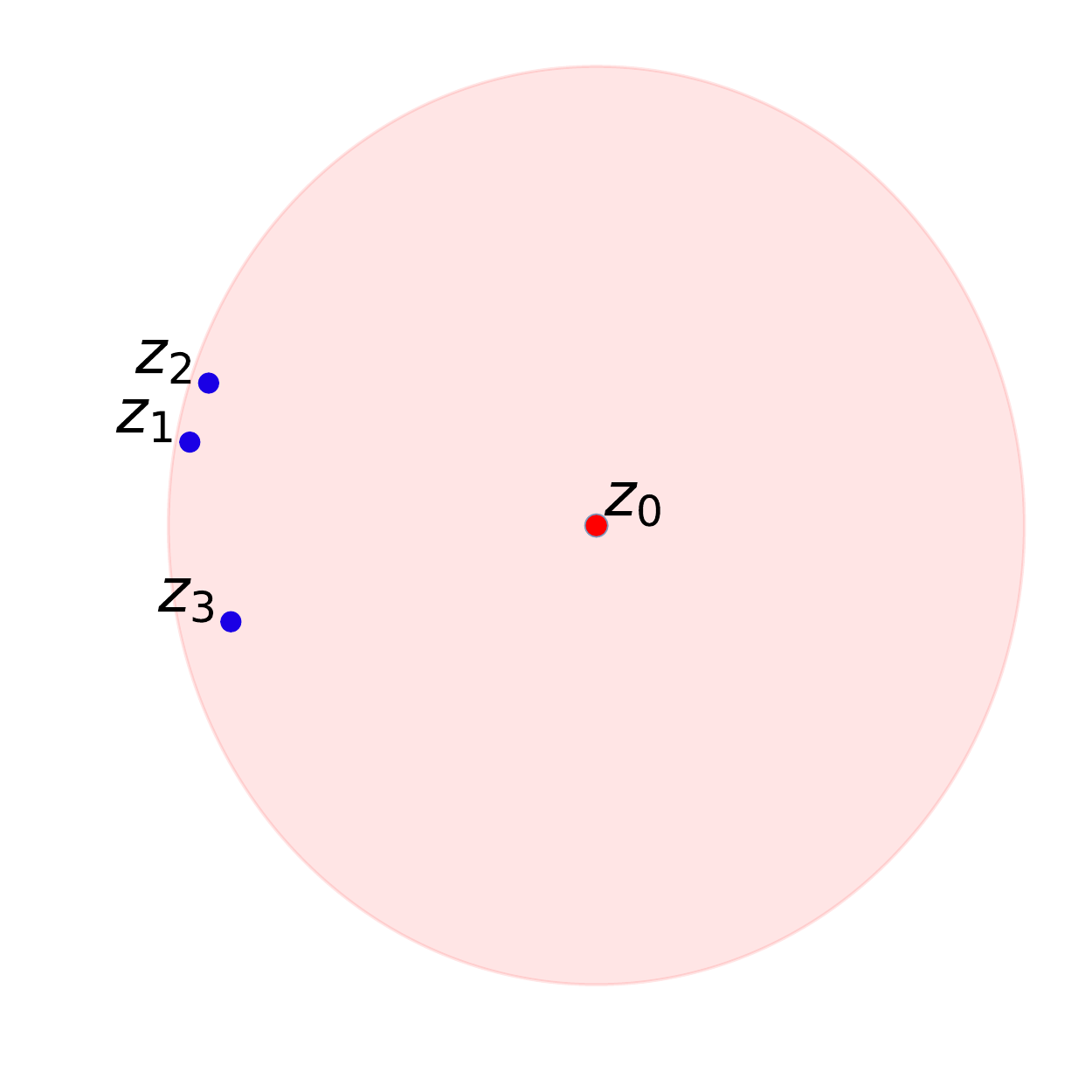}
         \caption{}\label{fig:leaved_embeddings}
     \end{subfigure}
    \begin{subfigure}[b]{0.19\textwidth}
         \centering
         \includegraphics[width=0.9\textwidth]{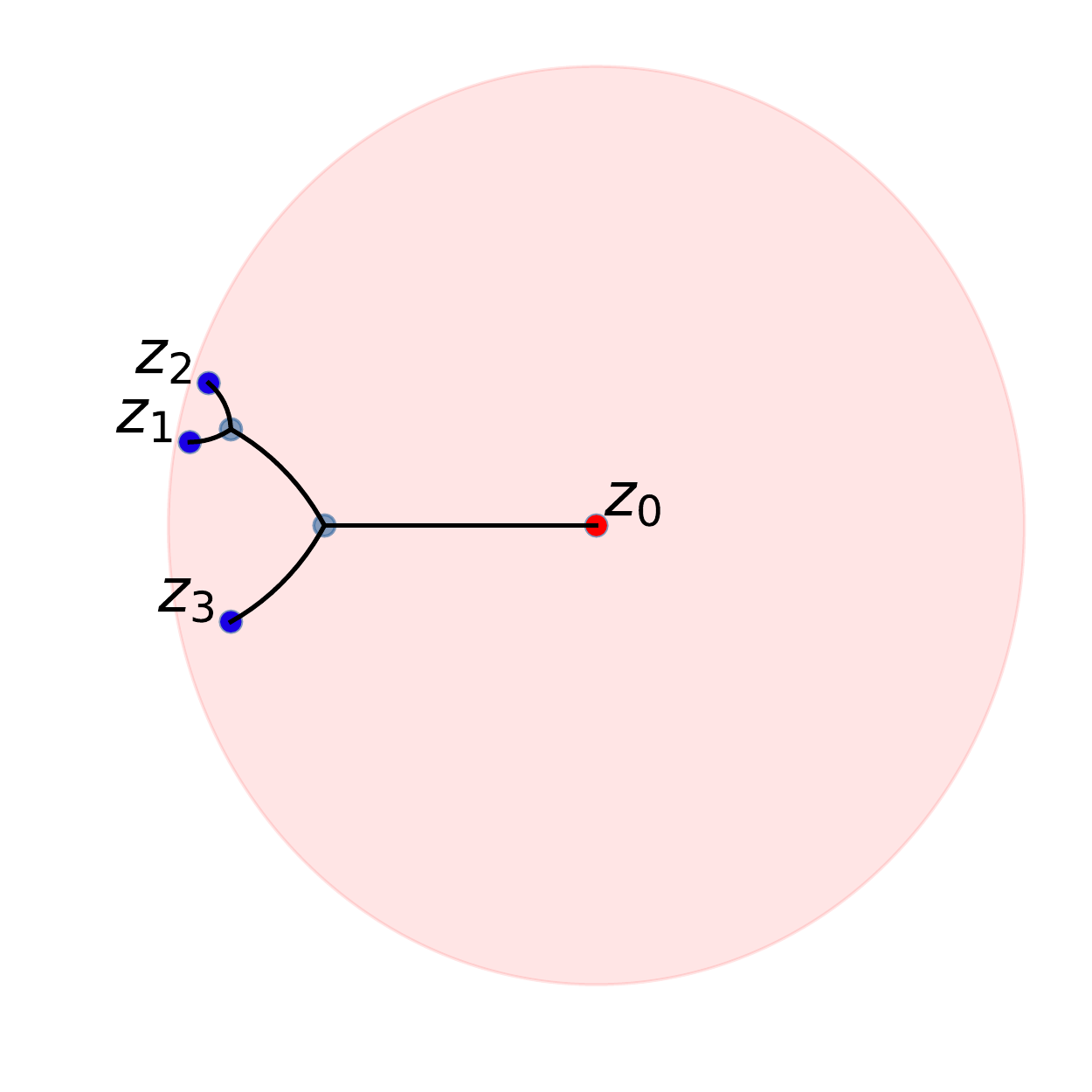}
         \caption{}\label{fig:quasi}
     \end{subfigure}
     \begin{subfigure}[b]{0.19\textwidth}
         \centering
         \begin{istgame}[font=\footnotesize]
        \xtdistance{7mm}{16mm}
        \istroot(0){$0$}
          \istb{}[l]{}[l]
          \endist
        \xtdistance{7mm}{10mm}
        \istroot(1)(0-1)<135>{$r$}[l]
          \istb{}[l]{$3$}[l]
          \istb{}[r]{}[r]
          \endist
        \xtdistance{7mm}{10mm}
        \istroot(2)(1-2)<45>{}
          \istb{}[l]{$2$}[l]
          \istb{}[r]{$1$}[r]
          \endist
        \end{istgame}
         \caption{}\label{fig:unrooted_tree}
     \end{subfigure}
        \begin{subfigure}[b]{0.19\textwidth}
         \centering
         \begin{istgame}[font=\footnotesize]
        \xtdistance{7mm}{10mm}
        \istroot(1){$0$}
          \istb{}[l]{$3$}[l]
          \istb{}[r]{}[r]
          \endist
        \xtdistance{7mm}{10mm}
        \istroot(2)(1-2)<45>{}
          \istb{}[l]{$2$}[l]
          \istb{}[r]{$1$}[r]
          \endist
        \end{istgame}
         \caption{}\label{fig:rooted_tree}
     \end{subfigure}
     \begin{subfigure}[b]{0.19\textwidth}
         \centering
         \resizebox{\textwidth}{!}{\renewcommand{\arraystretch}{0.95}
        \begin{tikzpicture}[
  inner/.style={fill = light blue,circle,draw,thick,minimum width=5mm,inner sep=0},
  small inner/.style={inner,minimum width = 3mm,yshift={1mm}},
  triangle/.style={fill = light blue,isosceles triangle,draw=,thick,shape border rotate=90,isosceles triangle stretches=true, minimum height=20mm,minimum width=15mm,inner sep=0,yshift={-5mm}},
  small triangle/.style={triangle, minimum height = 8mm, minimum width = 6mm },
  large triangle/.style={triangle,minimum width = 27mm,minimum height=36mm,yshift={-11mm}},
  very large triangle/.style={triangle,minimum width = 33mm,minimum height=44mm,yshift={-11mm}},
  level 1/.style={sibling distance=50mm},
  level 2/.style={sibling distance=25mm},
  level 3/.style={sibling distance=25mm},
  level 4/.style={sibling distance=25mm},
  level 4/.style={sibling distance=15mm},
  level 5/.style={sibling distance=7mm},
]
  \node[inner] {$0$}
     [child anchor=north]
    child {node[inner] {$r$}
        child {node[triangle,yshift={-3mm}] (a) {$T_0$}}
        child {node[triangle,yshift={-3mm}] (a) {$T_1$}}
        }
        ;
    \end{tikzpicture}}
         \caption{}\label{fig:decoding_proof}
     \end{subfigure}
    \caption{(a): Leaves embeddings and the origin in $\mathbb{B}_2$. (b): Binary tree that is a quasi-isometry of the leaves embeddings and the origin~(\cref{prop:treelike}) and corresponding discrete binary tree (not rooted) in (c). 
    (d): Rooted binary tree obtained by decoding hyperbolic leaves embeddings.
    Observe that the quasi-isometric tree in (c) and the decoded tree in (d) are LCA equivalent with respect to $0$.
    (e): Figure used in induction proof in~\cref{lemma:decoding}.
    }\label{fig:rooted_unrooted}
\end{figure}
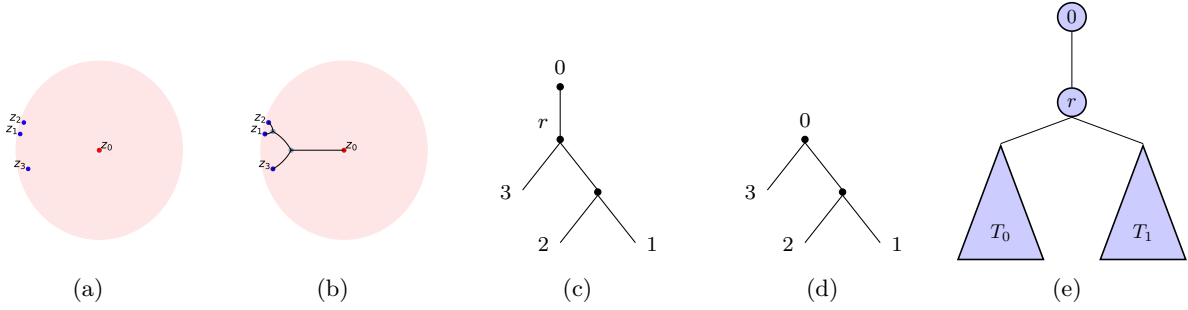

\subsubsection{Any spread embedding is in LCA agreement with some binary tree}\label{subsec:lca_agreement_condition}
\cref{lemma:decoding} shows that under the LCA agreement condition, the decoding algorithm will produce a discrete tree that preserves the ordering of triplets (in terms of LCAs' depths), thus the continuous and discrete costs are close.
Furthermore, we know by~\cref{prop:treelike} that any embedding has a corresponding quasi-isometric binary tree.  
Using this, we seek a condition under which a quasi-isometric pair $(Z, T)$ is in LCA agreement. 

First, we show that under a quasi-isometric embedding~(\cref{defn:quasi}), the hyperbolic LCA depth $d_o(z_i \vee z_j)$ is close to the tree LCA depth $d_0(i \vee j)$~(\cref{lmm:depth}). 
Next, we derive a condition under which a quasi-isometric pair $(Z, T)$ is in LCA agreement~(\cref{lemma:lca_agreement}).
Finally, we show that any spread embedding satisfies this condition, and more specifically, we show that any spread embedding is in LCA agreement with some binary tree, such that the ``deepest'' LCA is always distinguishable from the embeddings~(\cref{lmm:decode-condition}).

\begin{lemma}[Hyperbolic LCA depth is close to Tree LCA depth]%
    \label{lmm:depth}
    Let $Z\subset\mathbb{B}_2^{n+1}$ be an embedding with $z_0 = o$ and $T$ a binary tree (not rooted) on $(n+1)$ leaves such that $(Z, T)$ is $(1+\varepsilon, \kappa)$-quasi-isometric. 
    If $M$ is the diameter of $T$, then for any leaves $1\le i, j, k, l\le n$:
    \begin{equation}
        \label{eq:lca-depths}
        \left| (d_o(z_i \vee z_j) - d_o(z_k \vee z_l))
        - (d_0(i \vee j) - d_0(k \vee l)) \right|
        \le \frac{3\varepsilon}{2}M + \frac{3}{2}\kappa + \delta
        .
    \end{equation}
\end{lemma}
\begin{proof}%
    We bound the hyperbolic LCA depth by the corresponding tree depths:
     \begin{align*}
        d_o(z_i \vee z_j) &\ge \langle z_i, z_j \rangle
        \\
        &= \frac{1}{2}\left( d_o(z_i) + d_o(z_j) - d(z_i, z_j) \right)
        \\
        &\ge \frac{1}{2}\left( d_0(i) + d_0(j) - (1+\varepsilon)d(i, j) - \kappa \right)
        \\
        &= d_0(i \vee j) - \frac{\varepsilon}{2}d(i, j) - \frac{1}{2}\kappa
        .
    \end{align*}
    The first line is \cref{lmm:lca-gromov}, the second line is \cref{def:gromov}, and the third applies \cref{def:quasi}.

    In the other direction, again using \cref{lmm:lca-gromov}, \cref{def:gromov}, and \cref{def:quasi}:
        \begin{align*}
        d_o(z_i \vee z_j) &\le \langle z_i, z_j \rangle + \delta
        \\
        &= \frac{1}{2}\left( d_o(z_i) + d_o(z_j) - d(z_i, z_j) \right) + \delta
        \\
        &\le \frac{(1+\epsilon)\left( d_0(i) + d_0(j) \right) + 2\kappa}{2} - \frac{d(i, j)}{2} + \delta
        \\
        &\le \frac{1}{2}\left( d_0(i) + d_0(j) - d(i, j) \right) + \frac{\epsilon}{2}(d_0(i)+d_0(j)) + \kappa + \delta
        \\
        &= d_0(i\vee j)+ \frac{\varepsilon}{2}d_0(i) + \frac{\varepsilon}{2}d_0(j) + \kappa + \delta
        .
    \end{align*}
    Finally, \cref{eq:lca-depths} follows by adding these inequalities with $d(i, j) \le M$ for any $0 \le i, j \le n$.
\end{proof}

\cref{lmm:depth} allows us to provide a concrete condition for when $(Z, T)$ are in LCA agreement,
which will be our main tool for showing the consistency of our relaxation.
\begin{lemma}\label{lemma:lca_agreement}
    \label{lmm:isometry-agreement}
    Let $Z\subset\mathbb{B}_2^{n+1}$ be an embedding with $z_0 = o$ and $T$ a binary tree (not rooted) on $(n+1)$ leaves, such that $(Z, T)$ is $(1+\varepsilon, \kappa)$-quasi-isometric.
    Define $M$ to be the diameter of $T$ and $m$ to be the length of the smallest edge not including a leaf of $T$.
    If $m > \frac{3\varepsilon}{2} M + \frac{3}{2}\kappa + \delta$, then $(Z, T)$ are in LCA agreement.
\end{lemma}
\begin{proof}%
    Consider an arbitrary triple $(i, j, k)$, and suppose $\{ij | k\}_T$ holds.
    It suffices to show that $\{z_i, z_j | z_k\}_{\mathbb{B}_2}$ holds to show that $(Z, T)$ are in LCA agreement~(\cref{defn:lca_agreement}).
    Applying \cref{lmm:depth},
    \begin{align*}
        d_o(z_i \vee z_j) - d_o(z_i \vee z_k)
        &\ge 
        d_0(i \vee j) - d_0(i \vee k) - \left( \frac{3\varepsilon}{2}M + \frac{3}{2}\kappa + \delta \right)
        \\
        &\ge m - \frac{3\varepsilon}{2}M - \frac{3}{2}\kappa - \delta
        > 0
        .
    \end{align*}
    In the second line, we used the fact that since $i \vee j$ and $i \vee k$ are both internal nodes in $T$, their distance is at least $m$ by assumption (also, $i\vee j$ is deeper than $i \vee k$ by assumption, so the sign of the difference is positive).
    Similarly $d_o(z_i \vee z_j) - d_o(z_j \vee z_k) > 0$, so $z_i \vee z_j$ is the deepest out of the 3 LCAs, as desired.
\end{proof}

\cref{lmm:isometry-agreement} gives a technical condition for when an embedding $Z$ is in LCA agreement with a tree $T$.
Next, we claim that the constrained set of spread embeddings $\mathcal{Z} \subseteq \mathbb{B}_2^n$ is such that for every embedding $Z \in \mathcal{Z}$, there is always a corresponding tree $T$ satisfying \cref{lmm:isometry-agreement}.
Intuitively, \cref{lmm:isometry-agreement} says that $m$ should be large, i.e.\ the internal edges of the tree should be spread far apart.
Using our notion of hyperbolic LCA, we can codify this by enforcing that different LCAs should be far from each other.
By leveraging global properties of hyperbolic space, this is in fact sufficient.
\begin{lemma}%
    \label{lmm:decode-condition}
    Suppose that embedding $Z\in\mathcal{Z}\subset\mathbb{B}_2^n$ is spread.
    Then there exist a binary tree (not rooted) on $(n+1)$ leaves $T$ such that $(Z\cup\{z_0=o\}, T)$ are in LCA agreement.
    In particular, for any $1\le i, j, k\le n$ such that $\{i, j|k\}_T$ holds, then:
    \begin{align*}
        d_o(z_i\vee z_j)&>\mathrm{max}\{d_o(z_i\vee z_k), d_o(z_j\vee z_k)\} + 1.
    \end{align*}
\end{lemma}
\begin{proof}%
    Append the origin $z_0$ to $Z$ (so we have a collection of $n+1$ points).
    By \cref{prop:treelike}, there is a binary tree $T$ (not rooted) on $(n+1)$ leaves that is $(1, C_n)$-quasi-isometric for $Z\cup \{z_0\}$, where $C_n=\delta\cdot\mathcal{O}(n)$.
    Consider an internal edge $e \in T$, i.e. an edge connecting non-leaf nodes.
    Since $T$ is binary, $e$ has endpoints $i \vee j$, $i \vee j \vee k$ for some triplet $\{ij | k\}_T$ with $1\le i, j, k\le n$.\footnote{Note that triplets here are defined on leaves $[n]$, excluding $0$. The reason is that, since the LCA is computed with respect to $0$ which is a leaf in $T$, $0\vee i = 0\ \forall i$, and therefore there is no internal edge whose endpoint is an LCA on $0$.}

    Let $a, b, c$ be the ordering of $i, j, k$ such that $d_o(z_a \vee z_b) > d_o(z_b \vee z_c) > d_o(z_a \vee z_c)$.
    Consider:
    \begin{align*}
        d_0(a \vee b) - d_0(a \vee c)
        &\ge d_o(z_a \vee z_b) - d_o(z_a \vee z_c) -\left( \frac{3}{2}C_n + \delta \right) 
        \\
        &> \frac{3}{2}C_n + \delta + 1
        ,
    \end{align*}
    where the first line applies \cref{lmm:depth} with $\varepsilon=0, \kappa=C_n$, and the second uses the definition of spread~(\cref{def:spread-detailed}), which holds for any $1\le a, b, c\le n$.
    In particular, $d_0(a \vee b) - d_0(a \vee c) > 0$, so clearly $\{a,b|c\}_T$ holds since $T$ is binary, and we must have:
    \[
        d_0(a \vee b) - d_0(a \vee c) = d_0(i \vee j) - d_0(i \vee j \vee k) = d(i\vee j, i \vee j \vee k),
    \]
    which is the length of the edge $e$ we are considering.
    Since this holds generically for any edge $e$ among internal nodes of $T$, this also holds for the minimum edge length $m$:
    \begin{equation}\label{eq:min_edge}
        m >\frac{3}{2}C_n + \delta +1.
    \end{equation}
    In particular, the conditions of \cref{lmm:isometry-agreement} apply;
    that is, $m$ satisfies $m > \frac{3\varepsilon}{2}M + \frac{3}{2}\kappa + \delta$ for $\varepsilon=0, \kappa=C_n$.
    Applying~\cref{lmm:isometry-agreement}, we get that $(Z\cup\{z_0\}, T)$ are in LCA agreement.
    
    We now turn to the second part of~\cref{lmm:decode-condition} which bounds the difference in LCA depth. Let $1\le i, j, k\le n$ such that $\{i, j|k\}_T$ holds. Since $(Z, T)$ are in LCA agreement, we can assume WLOG that $d_o(z_i\vee z_j) > d_o(z_i\vee z_k) > d_o(z_j\vee z_k)$.
    Applying~\cref{lmm:depth} again, we have:
    \begin{align*}
        d_o(z_i\vee z_j)-d_o(z_i\vee z_k) & \ge d_0(i\vee j) - d_0(i\vee k) - (\frac{3}{2} C_n + \delta)\\
        &\ge m - \frac{3}{2}C_n - \delta\\
        &> 1,
    \end{align*}
    where we used~\cref{eq:min_edge} in the last inequality.
\end{proof}

\subsubsection{The tree decoded from any spread embedding has a cost close to the \name{} cost}\label{subsubsec:proof_decode}
We now have all the tools to show that any spread embedding decodes to a tree such that the discrete and continuous costs are close~(\cref{lmm:spread-equivalence}).
\decodelemma*
\begin{proof}
Let $Z\in\mathcal{Z}\subset\mathbb{B}_2^n$ be a spread embedding.
Using~\cref{lmm:decode-condition}, we know that there exists a binary tree (not rooted) on $(n+1)$ leaves $T$, such that $(Z\cup\{z_0=o\}, T)$ are in LCA agreement. 
Let $T'$ be the rooted binary tree that is obtained by removing the leaf $0$ from $T$, and relabelling its neighbor $r$ in $T$ as $0$, the root of $T'$.
Using~\cref{lemma:decoding}, we know that $T' = \dec(Z)$ and:
\begin{align*}
    |C_{\mathrm{Dasgupta}}(\dec(Z);w) - C_{\mathrm{\name{}}}(Z;w,\tau)|&\le\sum_{ijk}|w_{ijk}(T';w)-w_{\mathrm{\name{}},ijk}(Z;w,\tau)|.
\end{align*}
Let $\delta_{ijk}\coloneqq|w_{ijk}(T';w)-w_{\mathrm{\name{}},ijk}(Z;w,\tau)|$. 
WLOG, assume that $\{i, j|k\}_{T'}$ holds for a triplet $(i, j, k)\in T'$.
$T$ and $T'$ are equivalent in the LCA agreement sense with respect to $0$, since LCA agreement in defined over leaves in $[n]$. 
That is, the LCA of any pair $i,j\in[n]$ with respect to $0$ in $T$ is the same as the LCA with respect to $0$ in $T'$. 
Therefore using the definition of LCA agreement, we have:
$$d_1\coloneqq d_o(z_i\vee z_j )\ge \mathrm{max}\{d_o(z_i\vee z_k ), d_o(z_j\vee z_k )\}\coloneqq d_2.$$ 
Denote $w^*_{ijk}\coloneqq\mathrm{max}\{|w_{ij}|, |w_{ik}|, |w_{jk}|\}$ and $\Sigma_{ijk}\coloneqq e^{d_o(z_i\vee z_j )/\tau} + e^{d_o(z_i\vee z_k )/\tau} + e^{d_o(z_j\vee z_k )/\tau}$. 
Then:
\begin{equation}
    \begin{split}
        \delta_{ijk}&=\bigg|w_{ij}\bigg(1-\frac{e^{d_o(z_i\vee z_j )/\tau}}{\Sigma_{ijk}}\bigg) + w_{ik}\bigg(\frac{e^{d_o(z_i\vee z_k )/\tau}}{\Sigma_{ijk}}\bigg)+w_{jk}\bigg(\frac{e^{d_o(z_j\vee z_k )/\tau}}{\Sigma_{ijk}}\bigg)\bigg|\\
        &\le 2w^*_{ijk}\bigg(\frac{e^{d_o(z_i\vee z_k )/\tau} + e^{d_o(z_j\vee z_k )/\tau}}{\Sigma_{ijk}}\bigg)\\
        &\le 4 w^*_{ijk} e^{(d_{2}-d_1)/\tau}\\
        &\le 4 w^*_{ijk} e^{-1/\tau}.
    \end{split}
\end{equation}
In the last line, we applied the second part of~\cref{lmm:decode-condition}.
\end{proof}

\subsection{Proof of~\cref{lmm:encode-decode}}\label{sec:theory-encoding}
We have shown in~\cref{lmm:decode-condition} that any spread embedding gets decoded into a tree such that the continuous and discrete costs are close. 
We now show the other direction, that any rooted binary tree has a corresponding spread embedding which decodes back to it~(\cref{lmm:encode-decode}), and therefore the discrete and continuous costs are close.  

We first recall a result by Sarkar for low-distortion hyperbolic embeddings of trees~\cite{sarkar2011low}.
\begin{proposition}%
    \label{prop:sarkar}
    Any unit-weight tree $T$ can be embedded into $Z$ with scale $\zeta = O(1/\varepsilon)$  and worst-case distortion at most $1+\varepsilon$, i.e.
    \begin{align*}
        \zeta d_T(i, j)\le d(z_i,z_j)\le \zeta (1+\varepsilon)d_T(i, j).
    \end{align*}
\end{proposition}
In our terminology, this says that if every edge of $T$ is weighed with a scalar $\zeta=\mathcal{O}(1/\varepsilon)$, then there is an embedding $Z$ such that $(Z, T)$ is $(1+\varepsilon, 0)$-quasi-isometric, that is:
\begin{align*}
    d_T(i, j)\le d(z_i, z_j)\le (1+\varepsilon)d_T(i, j).
\end{align*}

We now rely on Sarkar's result to find a spread embedding for a given rooted binary tree.
\begin{lemma}%
    \label{lmm:encode}
    Let $T$ be any unit-weight binary tree (not rooted) on $(n+1)$ leaves.
    Then there is a spread embedding $Z\in\mathcal{Z}\subset\mathbb{B}_2^{n+1}$ with $z_0=0$, such that $(Z, T)$ are in LCA agreement.
\end{lemma}

\begin{proof}%
    Let $\varepsilon >0$.
    Put a weight $\zeta$ ($\zeta$ to be decided later) on every edge of $T$.
    Using \cref{prop:sarkar}, embed $T$ to an embedding $Z \subset\mathbb{B}_2^{n+1}$,
    such that $(Z, T)$ is $(1+\varepsilon, 0)$-quasi-isometric,
    and WLOG reflect the embeddings (isometric transformation) so that $z_0$ is at the hyperbolic origin.

    Consider any triplet $1\le i, j, k\le n$, and WLOG let $\{ij|k\}_T$. 
    Applying \cref{lmm:depth} with $\kappa=0$,
    \begin{align*}
        d_o(z_i \vee z_j) - d_o(z_i \vee z_k)
        &\ge d_0(i \vee j) - d_0(i \vee k) - \frac{3\varepsilon}{2}\zeta n - \delta
        \\
        &\ge \zeta \left(1 - \frac{3\varepsilon}{2} n\right) - \delta
    \end{align*}
    In the second line we used the fact that the diameter of $T$ is at most $\zeta n$ (since there are $n+1$ nodes and all edges are equally weighted),
    and in the third that $i \vee j \neq i \vee k$ and the minimum tree distance between any distinct nodes is $\zeta$.

    Finally, we could choose $\varepsilon \le \frac{1}{3n}$ and $\zeta > 6C_n + 6\delta + 2$.
    Note that choosing such $\zeta = \Theta(n)$ works since 
    \cref{prop:sarkar} says $\zeta = \Theta(1/\varepsilon) = \Theta(n)$ is possible for the embedding,
    and \cref{prop:treelike} says it is sufficient since $C_n = O(n)$.
    Then:
    \begin{align*}
        & \max \{d_o(z_i \vee z_j), d_o(z_i \vee z_k), d_o(z_j \vee z_k) \} - \min \{d_o(z_i \vee z_j), d_o(z_i \vee z_k), d_o(z_j \vee z_k) \}
        \\
        &\ge d_o(z_i \vee z_j) - d_o(z_i \vee z_k)
        \\
        &\ge \zeta \left(1 - \frac{3\varepsilon}{2} n\right) - \delta
        \\
        &> 3C_n + 2\delta + 1.
    \end{align*}
    Since $i, j, k$ were arbitrary, this shows that $Z$ is spread (as defined in~\cref{def:spread-detailed}).
    Finally, note that the min edge length of $T$ is $m=\zeta$, the maximum path length is at most $M=n\zeta$, and $\zeta > \frac{3}{2}\varepsilon (n\zeta) + \delta$ by choice of $\varepsilon$ and $\zeta$.
    Since $(Z, T)$ were a $(1+\epsilon, 0)$-quasi-isometry, \cref{lmm:isometry-agreement} then implies that $(Z, T)$ are in LCA agreement.
\end{proof}

Finally, we show~\cref{lmm:encode-decode} using the previous Lemma and~\cref{lemma:decoding}.
\encodelemma*
\begin{proof}
Attach a leaf $0$ to the root of $T$, apply \cref{lmm:encode}, then
apply \cref{lemma:decoding}.
\end{proof}

\section{Experimental details}
\begin{figure}[t]
    \centering
    \begin{subfigure}[b]{0.3\textwidth}
         \centering
         \includegraphics[width=0.9\textwidth]{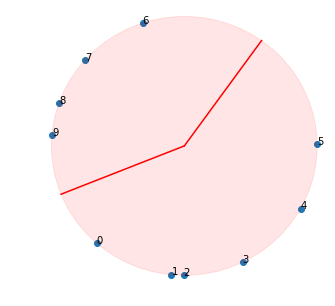}
         \caption{First split (red).}\label{fig:decoding_1}
     \end{subfigure}
    \begin{subfigure}[b]{0.3\textwidth}
         \centering
         \includegraphics[width=0.9\textwidth]{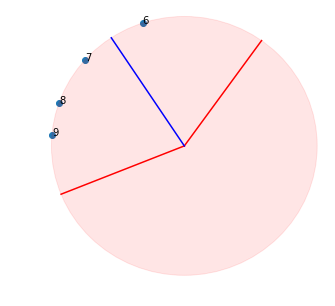}
         \caption{Second split (blue).}\label{fig:decoding_2}
     \end{subfigure}
     \begin{subfigure}[b]{0.3\textwidth}
         \centering
         \includegraphics[width=0.9\textwidth]{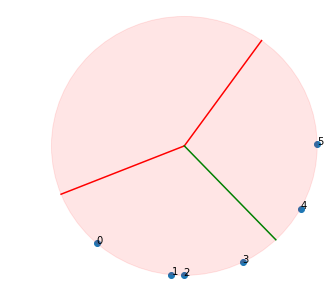}
         \caption{Third split (green).}\label{fig:decoding_3}
     \end{subfigure}
    \caption{An illustration of the greedy decoding algorithm.}\label{fig:angle_split}
\end{figure}
\begin{table}
  \centering
  \begin{tabular}{llccccc}
    \toprule
    & & \multicolumn{1}{c}{\textbf{Zoo}} & \multicolumn{1}{c}{\textbf{Iris}}
    &  \multicolumn{1}{c}{\textbf{Glass}} & \textbf{Segmentation} &  \multicolumn{1}{c}{\textbf{Spambase}} 
    \\
    \cmidrule(r){2-7}
     & \# Points & \multicolumn{1}{c}{101} & \multicolumn{1}{c}{150} 
     & \multicolumn{1}{c}{214} & 2310 & \multicolumn{1}{c}{4601} 
     \\
     & \# Clusters & 7 & 3 & 6 & 7 & 2     \\
    \midrule
    \multirow{5}{*}{Discrete} & SL & \underline{97.7 }& \underline{76.7} 
    & \underline{50.3} & 51.1 & 61.2\\
    & AL & 90.1 & 73.7 
    & 46.3 & 58.2 & {73.4}
    \\
    & CL & 96.6 & 76.1 & 46.9 & 55.1 & 69.1 
    \\
    & WL & 90.0 & 74.9 & 48.3 & \underline{61.3} & 68.3 
    \\
    \cmidrule(r){2-7} 
    & BKM & 86.5 & 66.1 & 43.5 & 57.2 & \underline{74.9} 
    \\
    \midrule
    \multirow{3}{*}{Continuous} & UFit & 97.2 & 76.8 & \textbf{51.0} & \textbf{61.2} & 61.6
    \\
    & gHHC & - & - & 46.3 & - & 61.4 
    \\
    \cmidrule(r){2-7}
    & \name{} & \textbf{98.7} & \textbf{77.3} & 49.2 & 55.9 & \textbf{79.2} 
    \\
    \bottomrule
    \end{tabular}
    \vspace{5pt}
    \captionof{table}{Clustering quality results measured in dendrogram purity (DP). 
    Best score in bold, second best underlined. gHHC scores are directly taken from~\cite{monath2017gradient}.}\label{table:den_purity}
\end{table}

\subsection{More experimental details}\label{subsec:appendix_exp_details}
\paragraph{Datasets}
In our experiments, we compute similarities using the cosine similarity measure on the datapoints' features. 
For all datasets in the UCI machine learning repository, we use the available features and normalize them so that each attribute has mean zero and standard deviation one. 
For CIFAR-100, where the raw data comes in the form of images, we used a pretrained BiT~\cite{kolesnikov2019large} convolutional neural network to compute 2048-dimensional image features (one to last layer). 

\paragraph{Baselines} To compute numbers for agglomerative clustering methods, we used the corresponding implementation in the scipy Python library.\footnote{\url{https://docs.scipy.org/doc/scipy/reference/cluster.hierarchy.html}}
We implemented our own version of BKM following the description in~\cite{moseley2017approximation} since no open-source version was available. 
For UFit, we used the open-source implementation and reused the same hyper-parameters as in the original paper~\cite{chierchia2019ultrametric}.\footnote{\url{https://github.com/PerretB/ultrametric-fitting}}
\paragraph{Evaluation metrics}
Dasgupta’s cost is a well-studied objective with known guarantees when there is an underlying ground-truth hierarchy~\cite{cohen2019hierarchical} (such results have not been established for metrics like dendrogram purity (DP)~\cite{heller2005bayesian}).
We therefore measure the clustering quality in terms of the discrete Dasgupta Cost.\footnote{For completeness, we also report DP scores in~\cref{table:den_purity}, observing that this metric does not correlate well with DC on some datasets (see~\cref{appendix:ghhc_comparison} for a discussion of the results).}
For randomized algorithms and our method (which produce a different solution for every run), we report the best cost over five random seeds.
This is standard since all methods can be viewed as \emph{search} algorithms for the latent minimizer of the Dasgupta cost,
analogous to how standard combinatorial search algorithms for NP-hard problems rely on random restarts and global perturbations when they reach local optima.

\paragraph{Cost Bounds}
We also report upper and lower bounds on the discrete Dasgupta cost, computed as:
\begin{equation}\label{eq:up_lb}
    \begin{split}
     \mathrm{UB}(w)&\coloneqq\sum_{ijk}\mathrm{max}(w_{ij}+w_{ik}, w_{ij}+w_{jk},w_{ik}+w_{jk}) + 2\sum_{ij}w_{ij}\ge C_\mathrm{Dasgupta}(T;w)\\
     \mathrm{LB}(w)&\coloneqq\sum_{ijk}\mathrm{min}(w_{ij}+w_{ik}, w_{ij}+w_{jk},w_{ik}+w_{jk}) + 2\sum_{ij}w_{ij}\le C_\mathrm{Dasgupta}(T;w),
    \end{split}
\end{equation}
for all rooted binary tree $T$.
For datasets with more than a thousand of nodes, we sample triplets uniformly at random for both the lower and upper bounds, and report the average over 10 random seeds.

Note that datasets where the relative gap between the upper and lower bound is larger indicate datasets where similarities induce a more hierarchical structure.
As noted in~\cite{wang2018improved}, an instance for which the lower bound can be achieved by a tree, is a ``perfect'' HC instance (termed ``perfect HC-structure'' in~\cite{wang2018improved}), in the sense that for every three points, the tree decomposes them in the most preferred way, i.e., by cutting the highest similarity weight last, towards the bottom of the tree. As the optimum tree gets higher costs, approaching the upper bound above, the instance loses its hierarchical structure; for example, if the given graph is a unit weight clique with no hierarchy to be found, both upper and lower bounds coincide and actually this implies that any tree gets the same cost, as was shown in Dasgupta~\cite{dasgupta2016cost}. 

In~\cref{table:small_exp}, we note that on dataset where there is a large relative gap between the upper and lower bounds (e.g. Iris or Segmentation), the relative improvement of \name{} compared to the best baseline is more important compared to datasets with a smaller gap (e.g. CIFAR-100). 
This suggest that when there is a good underlying HC in the data, \name{} is able to get closer to it than heuristic algorithms. 

\subsection{Greedy decoding}\label{appendix:greedy}
To provide more intuition about the greedy decoding, we illustrate different steps of greedy decoding on a small example in~\cref{fig:angle_split}.
The first split is computed using the two largest angles splits (red lines in~\cref{fig:decoding_1}). Then, the algorithm recurses on the two created subsets and uses the largest angle in each subset to split the data (blue and green lines in~\cref{fig:decoding_2} and~\cref{fig:decoding_3}).

\subsection{Comparison with gHHC}\label{appendix:ghhc_comparison}
\paragraph{Models' comparison}
For completeness, we discuss in more details the comparison between \name{} and the gHHC model. 
Monath et al. \cite{monath2019gradient} propose a differentiable HC objective, which yields improvements in scalability and downstream task performance.
This model learns representations for a fixed number of intermediate nodes using hyperbolic embeddings, and optimizes such embeddings for the HC task.
Once learned, the embeddings can be decoded into a discrete tree using heuristics and post-processing rules.
gHHC differs from the traditional similarity-based HC setting, assuming additional information about the optimal clustering, in the form of hyperbolic leaves' embeddings.
These are computed using normalized Euclidean features, which is mismatched to the geometry of the data.
In contrast, \name{} does not rely on input leaves' embeddings and directly optimizes the entire tree structure, via the hyperbolic LCA construction.
The \name{} decoding does not require post-processing and directly produces a dendrogram which matches the underlying geometry of the embeddings. 

\paragraph{Experiments}
While we compare to similarity-based HC methods in our main experiments, we also include a comparison to gHHC (which uses features) for completeness in~\cref{table:den_purity}. 
gHHC evaluates the clustering quality using the dendrogram purity (DP) measure~\cite{heller2005bayesian}. Given ground truth flat clusters, DP measures how well the clusters are preserved in the hierarchy.
Note that this metric will only be a good indicator for clustering quality when the ground truth flat clusters correlate with the hierarchy. A proxy to measure such a correlation is analyzing if methods that do well on DC also do well on DP. 
On Spambase, \name{} and BKM have the best DC and DP scores, while SL does poorly for both metrics. 
This suggests that ground truth flat clusters do correlate well with the optimal hierarchy on this dataset, and we note that \name{} significantly outperforms gHHC for the DP metric. 
On Glass, the correlation is not so obvious; for instance SL performs poorly on DC but does very well on DP. 
We conjecture that the ground truth flat clusters in glass are not directly correlated with the optimal clustering on this dataset.

\end{document}